\documentclass{article}
\def\notes{1}

\usepackage{amsmath}
\usepackage{amssymb}
\usepackage{algorithm}
\usepackage{algorithmicx}
\usepackage[noend]{algpseudocode}
\usepackage{amsthm}
\usepackage{macros}
\usepackage{fullpage}
\usepackage{thmtools}
\usepackage{dsfont}

\newcommand{\ErrInd}{\mathrm{ErrorIndicator}}
\newcommand{\RobInd}{\mathrm{RobustIndicator}}
\newcommand{\test}{\mathrm{test}}
\newcommand{\iso}{\mathrm{iso}}

 \usepackage[preprint]{neurips_2025}

\usepackage[utf8]{inputenc} 
\usepackage[T1]{fontenc}    
\usepackage{hyperref}       
\usepackage{url}            
\usepackage{booktabs}       
\usepackage{amsfonts}       
\usepackage{nicefrac}       
\usepackage{microtype}      
\usepackage{xcolor}         

\title{Robust learning of halfspaces under log-concave marginals}

\author{%
  Jane Lange\thanks{Alphabetical order.} \\
  MIT\\
  \texttt{jlange@mit.edu} \\
   \And
   Arsen Vasilyan$^*$ \\
   UT Austin \\
   \texttt{ArsenVasilyan@gmail.com} \\
}

\begin{document}

\maketitle

\begin{abstract}

We say that a classifier is \emph{adversarially robust} to perturbations of norm $r$ if,
with high probability over a point $x$ drawn from the input distribution, there is no point within distance $\le r$ from $x$ that is classified differently. 
The \emph{boundary volume} is the probability that a point falls within distance $r$ of a point with a different label.
This work studies the task of computationally efficient learning of hypotheses with small boundary volume, where the input is distributed as a subgaussian isotropic log-concave distribution over $\mathbb{R}^d$.

Linear threshold functions are adversarially robust; they have boundary volume proportional to $r$.
Such concept classes are efficiently learnable by polynomial
regression, which produces a polynomial threshold function (PTF), but PTFs in general may have boundary volume $\Omega(1)$, even for $r \ll 1$.

We give an algorithm that agnostically learns linear threshold functions and returns a classfier with boundary volume $O(r+\varepsilon)$ at
radius of perturbation $r$.
The time and sample complexity of $d^{\tilde{O}(1/\varepsilon^2)}$ matches the 
complexity of polynomial regression.

Our algorithm augments the classic approach of polynomial regression with three additional steps:
\begin{itemize}
		\item[a)] performing the $\ell_1$-error regression under noise sensitivity constraints,
		\item[b)] a structured partitioning and rounding step that returns a Boolean classifier with error $\opt + O(\varepsilon)$ 
				and noise sensitivity $O(r+\varepsilon)$ simultaneously, and
		\item[c)] a local corrector that ``smooths'' a function with low noise sensitivity into a function that is adversarially robust.
\end{itemize}
\end{abstract}

\section{Introduction}
\label{sec: intro}
A predictor is robust to adversarial examples if for most possible inputs,
a small perturbation will not cause the input to be misclassified.
We define the \emph{boundary volume} as the probability, over the input distribution, that a point is close to the boundary:
\[\AdvRob_{\mcD, r}(f) = \Prx_{x \sim \mcD}[\exists z: \norm{z}_2 \le r \text{ and } f(x) \ne f(x +z)].\]
The sum of the boundary volume and the classification error is a natural upper bound on the standard notion of robust risk:
\[\mathrm{RobRisk}_{\mcD, r}(f) = \Prx_{(x,y) \sim \mcD}[\exists z: \norm{z}_2 \le r \text{ and } y \ne f(x +z)].\]
In this work, we discuss computationally efficient learning of classifiers with optimal classification error and boundary volume, thus minimizing the robust risk.

There is a large body of research on adversarial robustness in machine learning,
the focus of which is to assess the robustness of classifiers commonly deployed in practice --- see \cite{BiggioCMNSLGR13, Goodfellow2014ExplainingAH,SzegedyZSBEGF13}, as well as \cite{kolter18} for an overview of the topic. It is well known that deep networks trained on classic benchmark data sets, such as ImageNet, can be ``tricked'' into misclassifying a test input by making a perturbation so small that it cannot be detected by humans, and that training robust models is a challenge in practice.

In this work we consider
linear threshold functions (halfspaces), one of the most basic and well-known concept classes. As 
a robustness benchmark, we consider the robustness of a \emph{proper} learning algorithm --- one that outputs a hypothesis that is a halfspace. 
For a halfspace, when the input is distributed according the $d$-dimensional standard Gaussian --- or more generally, a subgaussian isotropic log-concave distribution --- the probability that an input falls within Euclidean distance $r$ of the classification boundary is $O(r)$; thus a $1 - O(r)$ fraction of inputs are robust to adversarial perturbations of norm $r$.
A proper agnostic learner would output a hypothesis with the following guarantees:
\begin{itemize}
\item \textbf{Agnostic approximation:} $\Prx_{(x,y) \sim \mcD}[y \ne h(x)] \le \opt + \eps$, where $\opt$ is the classification error of the best halfspace, and
\item \textbf{Adversarial robustness:} $\AdvRob_{\mcD, r}(h) \le O(r)$.
\end{itemize}

However, the state of the art for agnostically learning halfspaces is an improper algorithm --- an algorithm that does not output a halfspace, but instead outputs a polynomial threshold function (PTF).
Polynomial regression with randomized rounding learns halfspaces over $\mcN(0,I_d)$ with time and sample complexity $d^{\tilde{O}(1/\eps^2)}$ and satisfies the accuracy guarantee \cite{diakonikolas_bounded_2009, kalai2008agnostically}.
In fact, the efficiency of polynomial regression is due to the inherent robustness of halfspaces --- their small surface area! 
However, there exist degree-$1/\eps^2$ PTFs with boundary volume $\Omega(1)$ even at a small radius of perturbation,\footnote{Consider the PTF $\sign(\prod_{i=1}^{1/\eps^2} x_i)$, which has boundary volume $\Omega(1)$ over $\mcN(0, I_d)$ at any radius $\ge \eps^2$.}
and this algorithm makes no guarantee that its output is robust.

No proper agnostic learning algorithm with $2^{o(d)}$ running time for general subgaussian isotropic log-concave distributions is known,
and for the Gaussian distribution the state of the art for proper learning is $d^{O(1/\eps^4)}$ \cite{diakonikolas21b}. 
The complexity of proper learning in these settings is still an open question.
We circumvent the possible difficulty of proper learning by giving an improper algorithm that still satisfies both the accuracy and robustness properties.
Our algorithm has time and sample complexity $d^{\tilde{O}(1/\eps^2)}$,
matching the run-time of the best (improper and non-robust) agnostic learning algorithm \cite{diakonikolas_bounded_2009}. 
Up to polylogarithmic factors in the exponent, this also matches the statistical query lower bound for agnostically learning halfspaces \cite{diakonikolas2021optimality}.

\subsection{Main result}
Our main result is an agnostic learner (\Cref{alg:RobustLearn}) for halfspaces that outputs an adversarially robust hypothesis.

\begin{theorem}[Adversarially robust learning of halfspaces, informal]
\label{thm:main-informal}
Let $\mcD$ be a distribution over $\R^d \times \bits$ such that the $\R^d$-marginal is subgaussian, isotropic, and log-concave.
There is an algorithm that takes a robustness radius $r$, an accuracy parameter $\eps$, and a sample of size 
$d^{\tilde{O}(1/\eps^2)}$ from $\mcD$.
It outputs a hypothesis $h: \R^d \to \bits$ with the following guarantees:
\begin{itemize}
		\item \textbf{Agnostic approximation:} $\Prx_{(x,y) \sim \mcD}[h(x) \ne y] \le \opt + O(\eps)$, 
				where $\opt$ is the classification error of the best halfspace.
\item \textbf{Adversarial robustness:} $\AdvRob_{\mcD, r}(h) \le O(r)$.
\end{itemize}
\end{theorem}

\subsection{Technical overview and intermediate results}
Our algorithm and its analysis have three main components, which may be of independent interest.
Two of the components solve relaxations of the robust learning problem, and the third transforms the almost-robust hypothesis into a robust one.
The relaxed learners produce hypotheses with small \emph{noise sensitivity} (a notion of boundary volume with random perturbations instead of adversarial perturbations)
and small \emph{isolation probability} (the probability that the local noise sensitivity around a point is very high). 
These quantities are defined formally in \Cref{sec:perturbations}.

\subsubsection{Local correction of adversarial robustness (\Cref{alg: RobustnessLCA})}
The part of our algorithm that transforms an almost-robust hypothesis into a robust one is a \emph{local corrector} for adversarial robustness. 
Local correctors are part of the family of \emph{local computation algorithms,} or LCAs \cite{rubinfeld2011fast, alon2012space}: fast randomized algorithms that compute parts of a large object without constructing the object in its entirety.
Local correctors, also known as local property reconstructors \cite{ailon2008property}, are LCAs that evaluate queries to a nearby object that satisfies a desired property; in our case that object is an adversarially robust hypothesis. 
Our algorithm uses local correction in the fashion of \cite{lange2022properly}: to ``append'' a local corrector with a fixed random seed to a non-robust hypothesis in order to \emph{globally} correct the hypothesis.

We analyze a very simple LCA that makes queries to
a function and outputs queries to a nearby function with reduced boundary volume.
This algorithm estimates the probability that a random perturbation of $x$
causes the label to change, and changes the label if this probability is too high.
This is, in essence, the smoothing procedure discussed in \cite{LAGHJ19, li2018second, cohen2019certified}.
We give guarantees on the boundary volume and the error introduced by the smoothing procedure in terms of the relaxed robustness properties of the input function (the noise sensitivity and isolation probability).

\begin{lemma}[Local corrector for adversarial robustness, informal]
Let $\eps, \alpha, \beta, r \in (0,1)$. Let $f:\R^d \to \bits$ be a degree-$k$ polynomial threshold function with noise sensitivity $\le \alpha$ and isolation probability $\le \beta$. There is an efficient randomized
algorithm \textsc{RobustnessLCA} that makes black-box queries to $f$ and answers black-box queries to a function
$g: \R^d \to \bits$ that satisfies the following:
\begin{itemize}
		\item \textbf{Adversarial robustness:} $\AdvRob_{\mcD, r} \le O(\alpha + \eps)$.
		\item \textbf{Small distance:} $\Prx_{x \sim \mcD}[g(x) \ne f(x)] \le O(\beta + \eps)$.
\end{itemize}
\end{lemma}

We use this algorithm in a ``deterministic'' fashion, where all calls to the local corrector share the same random seed, which is good with high probability. 
When the random seed is good, the guarantees hold for all $x \in \R^d$.\footnote{We remark that in a setting where each call to the corrector is allowed to fail \emph{independently} with probability $\delta$, the query complexity can be reduced from $d^{O(k)} \cdot \log(1/\delta)/\eps^2$ to $O(\log(1/\delta)/\eps^2)$ and the assumption that $f$ is a degree-$k$ PTF can be dropped.
This is not applicable in our setting but might be desired in a non-learning application of local correction for robustness.}
This allows us to append the local corrector and its fixed seed to some degree-$k$ PTF, 
creating a deterministic robust hypothesis that can be evaluated in $d^{O(k)}$ time.

From the guarantees of the local corrector and the desired guarantees of the overall algorithm,
we determine that we must feed the local corrector a PTF of error $\opt + O(\eps)$, noise sensitivity $O(r)$, and isolation probability $O(\eps)$.

\subsubsection{Polynomial regression under noise sensitivity constraints (\Cref{alg: LearnRealValued})}
Our first step in finding such a PTF is to learn a polynomial with low error, noise sensitivity, and isolation probability in the $\ell_1$-distance regime.

\begin{theorem}[Learning a polynomial, informal]
\label{thm:real-valued-informal}
There is an algorithm with time and sample complexity $d^{\tilde{O}(1/\eps^2)}$ that takes samples from the distribution $\mcD$ and returns a degree-$\tilde{O}(1/\eps^2)$ polynomial $p$ with the following properties:
\begin{itemize}
\item \textbf{Accuracy:} $\err(p) \le \opt + O(\eps)$,
\item \textbf{Low noise sensitivity:} $\NS(p) \le O(r + \eps)$, and
\item \textbf{Low isolation probability:} $\iso(p) \le O(\eps)$.
\end{itemize}
\end{theorem}
We achieve this with convex programming.
The $\ell_1$ noise sensitivity is a convex constraint, and we use a convex upper bound on the $\ell_1$ analogue of the isolation probability.
We minimize error over the set of degree-$\tilde{O}(1/\eps^2)$ polynomials under these constraints.
To show that a feasible polynomial of error $\opt + O(\eps)$ exists, 
we show that the halfspace-approximating polynomial given in \cite{diakonikolas_bounded_2009} satisfies the constraints
under any subgaussian isotropic log-concave distribution.
\subsubsection{Randomized partitioning and rounding (\Cref{alg:ComputeClassifier}, \Cref{alg:ComputeRoundingThresholds})}
With $p$ in hand, we must then round $p$ to some polynomial threshold function $\sign(p-t)$ that satisfies the error, noise sensitivity, and isolation probability constraints.
Simply rounding at a uniformly random threshold $t \sim [-1,1]$, as in \cite{kalai2008agnostically}, does result in error $\opt + O(\eps)$, noise sensitivity $O(r)$, and isolation probability $O(\eps)$ in expectation,
but doesn't guarantee that the conditions ever hold simultaneously for the same $t$.
Thus, running the local corrector on just one rounded function would not guarantee a hypothesis with the right properties.
In this step, we find a ``deterministic mixture'' of rounded functions that simultaneously satisfies the three conditions: a partition of the domain into parts, where each part is rounded at a different threshold.
The corrector is run on each part separately.

A first attempt at finding such a partition would be to allow it to have $O(1/\eps^2)$ parts.
Observe that the error, noise sensitivity, and isolation probability concentrate with deviation $\eps$ when averaging over $O(1/\eps^2)$ random thresholds. 
Thus, an equal-weighted mixture of $O(1/\eps^2)$ independent random roundings would satisfy all of the guarantees simultaneously with high probability.
But suppose we partitioned the domain into $1/\eps^2$ sets of equal mass --- we can't guarantee robustness for any point near the boundaries of these sets, and the total volume of these boundaries scales inversely with $\eps$, which is undesirable. 
To minimize the increase in boundary volume due to partitioning, it is necessary to find a partition of constant size.

We show that there exists a mixture of four rounded functions satisfying all the constraints simultaneously by applying Carath\'eodory's theorem (\cite{Caratheodory1907}),
and we find the mixture by linear programming.

To understand how to turn the mixture into a partition, consider the example of handling just the error when $y$ is a deterministic function of $x$. 
The error of the mixture is
\[\sum_{i \in [4]} w_i \cdot \err(\sign(p-t_i)) \le \opt + O(\eps),\]
where $w_i$ is the mixing weight of the $i^{th}$ PTF.
For each $i$ in the mixture, there is a set of volume $\err(\sign(p-t_i))$ consisting of the points misclassified by $\sign(p-t_i)$. We want to partition the domain such that a $w_i$-fraction of this set falls into the $i^{th}$ part, so that
\[\sum_{i \in [4]}  \Ind[x\text{ is in the $i^{th}$ part and }y \ne \sign(p(x) - t_i)] = \sum_{i \in [4]} w_i \cdot \err(\sign(p-t_i)).\]
We cite a theorem of \cite{dasgupta2006concentration}, which says, essentially, that projecting on a random unit vector causes the set to be very close to normally-distributed with high probability.
Thus, our partition is according to the inner product with a random unit vector $u$, and the parts are intervals $J_1,\ldots,J_4$ of Gaussian mass $w_1,\ldots,w_4$ respectively.
With high probability, the sets of misclassified points and the sets for which robustness is guaranteed by the corrector are all partitioned with the correct weights.
\begin{theorem}[Randomized partitioning, informal]
Let $p$ be a polynomial satisfying the guarantees of \Cref{thm:real-valued-informal}. 
There is an algorithm that outputs a unit vector $u$, rounding thresholds $t_1, \ldots, t_4$, and a partition of the real line into intervals $J_1,\ldots J_4$ such that the hypothesis
\[h(x) = \sum_{i =1}^4 \textsc{RobustnessLCA}(\sign(p(x) - t_i)) \cdot \Ind[\langle x, u \rangle \in J_i]\]
has the following properties:
\begin{itemize}
\item \textbf{Accuracy:} $\err(h) \le \opt + O(\eps)$,
\item \textbf{Robustness:} $\AdvRob_{\mcD, r}(h) \le O(r + \eps)$.
\end{itemize}
The time complexity is $d^{\tilde{O}(1/\eps^2)}$.
\end{theorem}
\subsection{Verifiable robustness}
\label{sec:intro-verifiable}
The work of \cite{goldwasser2022planting} discusses the planting of ``backdoors'' in learned classifiers --- structured violations of the robustness condition --- and the impossibility of efficiently distinguishing a robust model from a backdoored one in the most general case. 
The takeaway is that when training is performed by an untrusted service, it is not generally possible to verify that a hypothesis is robust, even if a description of the hypothesis is available.
However, our learning algorithm produces a hypothesis with a specific structure that can be checked, and the reduced expressivity of this structure allows robustness to be verifiable.
Under complexity assumptions, the robustness of our algorithm's output can be verified in the following way:

\begin{corollary}[Deterministic robustness, informal]
If $\textsf{P} = \textsf{BPP}$, then there is a learning algorithm $\mathcal{B}$ that, given access to labeled samples from a subgaussian isotropic log-concave distribution, runs in time $d^{\tilde{O}(1/\eps^2)}$ and produces a hypothesis $h$ with the following guarantees:
\begin{itemize}
\item \textbf{Agnostic approximation:} $\Prx_{(x,y) \sim \mcD}[h(x) \ne y] \le \opt + O(\eps)$, 
				where $\opt$ is the misclassification error of the best halfspace.

\item \textbf{Verifiable robustness:} There is an efficient deterministic verifier that takes as input a hypothesis $g$ and a point $x \in \R^d$, and always rejects if 
\[\exists z:\norm{z}_2 \le r \text{ and } g(x) \ne g(x+z).\]
If $g$ is the output of $\mathcal{B}$, then the verifier accepts with probability at least $1 - O(r + \eps)$ over $x$ drawn from $\mcD$.
\end{itemize}
\end{corollary}

As a result, training can be done entirely by an untrusted service who claims to be using our learning algorithm, but may not be.
The verifier checks whether the hypothesis matches the format our learning algorithm produces, rejects if it doesn't match, and performs a robustness test that is sound if it does match.
The user must trust that the error of the hypothesis really is $\opt + O(\eps)$, but
does not have to place any trust in the service to guarantee robustness --- there is a deterministic algorithm that the user can perform to verify that most points are not near the boundary.
Furthermore, in a setting where the user performs some of the construction of the hypothesis (but need not access any training data),
verifiable robustness can be made unconditional, but the soundness of the verifier can fail with small probability over the randomness of the hypothesis construction.
See \Cref{sec:verifiable} for further discussion.

\subsection{Related work.}
\noindent \textbf{Adversarially robust learning.}
Several recent papers study adversarially robust learning of various concept classes in the distribution-free PAC setting.
These papers also generalize the set of adversarial perturbations beyond the $\ell_2$
ball. 
Some consider $\ell_p$ perturbations and some consider fully arbitrary sets.

The works of \cite{cullina2018pac, montasser2019vc} show that there is essentially no statistical separation between robust learning and standard learning;
thus any gap in the hardness of these tasks must be computational. 
The works of \cite{bubeck2019adversarial, degwekar2019computational, AwasthiDV19} exhibit concept classes and perturbation sets such that, 
under standard hardness assumptions,
robust learning is computationally harder than standard learning.
The work of \cite{shafahi2018adversarial} exhibits some concept classes and perturbation sets such that a high robust risk is inevitable.

The work of \cite{montasser2020efficiently} gives an algorithm for robustly learning halfspaces in the \emph{realizable} setting where the data is assumed to be labeled by a halfspace with random classification noise.
The works of \cite{DKM20, DiakonikolasKM19} give algorithms for robust, proper, agnostic learning of halfspaces; however, the data distributions are assumed to be supported on the unit ball.
When scaled up to match the parameters of our setting --- distributions concentrated on a radius-$\sqrt{d}$ ball --- 
the running times of these algorithms have an exponential dependence on $\sqrt{d}$.

The work of \cite{AwasthiDV19} also studies robust learning in the distribution-free setting.
They give an algorithm for learning halfspaces and degree-2 PTFs with robustness under $\ell_\infty$-bounded perturbations.

In the distribution-specific setting,
\cite{gourdeau2021hardness} gives an algorithm for learning monotone conjunctions with robustness to perturbations of $O(\log n)$ Hamming distance, under log-Lipschitz distributions on the Hamming cube. 
As mentioned earlier, the proper agnostic learner of \cite{diakonikolas21b} is a robust learner for halfspaces under the Gaussian distribution.

\noindent \textbf{Other work on learning halfspaces.} See \Cref{sec: intro} for comparison of our work with \cite{diakonikolas_bounded_2009, kalai2008agnostically}.
Under general log-concave distributions, \cite{awasthi2017power} gives an algorithm for semiagnostic proper learning of origin-centered halfspaces under log-concave distributions, i.e. the algorithm gives a halfspace with error $O(\opt)+\eps$, where $\opt$ is the error of best halfspace. In contrast to this, our work focuses on obtaining the optimal $\opt+\eps$ error. \cite{daniely2015ptas} gives a poly-time method of achieving error $(1+\alpha)\opt+\eps$ under for any constant $\alpha$, but via an improper learning algorithm that uses polynomial regression\footnote{The algorithm of \cite{daniely2015ptas} also applies only to the Gaussian distribution, and not to genera}. The work of \cite{diakonikolas21b}, in addition to the $d^{\tilde{O}(1/\epsilon^4)}$-time algorithm for achieving error $\opt+\eps$ for halfspaces under the Gaussian distribution (discussed in \Cref{sec: intro}), gives a poly-time proper learning algorithm for origin-centered halfspaces that achieves error $(1+\alpha)\opt+\eps$ (similar to \cite{daniely2015ptas}). All algorithms in \cite{diakonikolas21b} are highly specific to Gaussian distributions and do not extend to, for example, to the uniform distribution over $[-1,1]^d$ and other sub-Gaussian log-concave distributions.

\section{Preliminaries}

When dealing with a distribution $\mcD$ over $\R^d \times \{\pm 1\}$, we will sometimes overload the notation to use the symbol $\mcD$ to refer to the $\R^d$-marginal of $\mcD$. 
Analogously, when dealing with a collection $S$ of pairs $\{(x_i, y_i)\}$ in $\R^d \times \{\pm 1\}$ we will sometimes write $\Pr_{x \sim S}[\cdots]$ to refer to $\Pr_{(x,y) \sim S}[\cdots]$ when the values of $y$ are not referenced.
\subsection{Perturbation and robustness models}
\label{sec:perturbations}
\label{sec: definitions of phi and psi}
\begin{definition}
[$r$-boundary volume]
\label{def:RR}
The boundary volume of a function $f$ at radius $r$ on distribution $\mcD$ is defined as 
\[\AdvRob_{\mcD,r}(f) = \Prx_{x \sim \mcD} [\exists z: \norm{z}_2 \le r\text{ and } f(x) \ne f(x+z)].\]
\end{definition}

\begin{definition}[Local noise sensitivity]
\label{def:local-NS}
We define the local noise sensitivity $\phi_{f,\eta}(x)$ at noise scale $\eta$ as 
\[\phi_{f,\eta}(x) \coloneqq \Ex_{z \sim \mcN(0,I_d)}[\lfrac 12 |f(x) - f(x + \eta z)|].\]
We remark that when $f$ is Boolean-valued, this is equivalently
\[\phi_{f,\eta}(x) \coloneqq \Prx_{z \sim \mcN(0,I_d)}[f(x) \ne f(x + \eta z)].\]
\end{definition}
\begin{definition}[Distributional noise sensitivity
\footnote{
We remark that when $\mcD = \mcN(0,I_d)$, our definition of distributional noise sensitivity differs from the standard definition of Gaussian noise sensitivity.
}]
\label{def:NS}
The noise sensitivity of a function $f$ at noise scale $\eta$ on distribution $\mcD$ is defined as 
\[\NS_{\mcD, \eta}(f) = \Ex_{x \sim \mcD}[\phi_{f,\eta}(x)].\]
\end{definition}

\begin{definition}[Isolation probability]
We call a point \emph{isolated} if its local noise sensitivity is over a threshold.
The \emph{isolation probability} of a function $f$ relative to a threshold $t$ is
\[\iso_{\mcD, \eta}(f, t) := \Prx_{x \sim \mcD}[\phi_{f, \eta}(x) > t].\]
We also use a convex relaxation of the isolation indicator and define $\psi$ to be its expectation:
\[\psi_{\mcD, \eta}(f) := 10 \Ex_{x \sim \mcD}[\Ind[\phi_{f,\eta}(x) > 0.6] \cdot ( \phi_{f, \eta}(x) - 0.6)].\]
Note that $\psi(f)$ is an upper bound on $\iso(f, 0.7)$.
\end{definition}


\subsubsection{Noise sensitivity approximators}
\label{sec: noise sensitivity approximators}

Below, we introduce the notion of a local noise sensitivity approximator. Many of the algorithms we introduce will use such an approximator as a black box they can query (see \Cref{alg:RobustLearn}, \Cref{alg:ComputeClassifier}, \Cref{alg:ComputeRoundingThresholds} and \Cref{alg: RobustnessLCA}). We first introduce a local noise sensitivity approximator and some related notions, and present a randomized method for efficiently instantiating it.\footnote{In \Cref{sec: randomized construction of phi hat} a randomized construction of $\hat{\phi}$ is given and in \Cref{sec:verifiable} a deterministic method is presented assuming $\mathsf{BPP} = \mathsf{P}$.}
{
\begin{definition}[Local noise sensitivity approximator for PTFs]
\label{def: local noise sensitivity approximator}
For a degree parameter $k$, an algorithm $\hat{\phi}$ is an $\epsilon$-accurate \emph{noise sensitivity approximator} if whenever it is given (i) a degree-$k$ polynomial threshold function $f$, (ii) a noise rate $\eta \in (0,1)$, and (iii) a point $x \in \R^d$, it outputs a value $\hat{\phi}_{f,\eta}(x)$ such that 
\[
\left\lvert
\hat{\phi}_{f, \eta}(x) -
\phi_{f, \eta}(x)
\right\rvert
\leq \epsilon.
\]
\end{definition}

\begin{definition}
Let $\hat{\phi}$ be as in the definition above.
Then we define the empirical noise sensitivity and isolation probabilities as 
\[\wh{\NS}_{\mcD, \eta}(f) \coloneqq  \Ex_{x \sim \mcD} \hat{\phi}_{f,\eta}(x) \quad \text{and} \quad \wh{\iso}_{\mcD, \eta}(f, t) \coloneqq \Ex_{x \sim \mcD} \Ind[\hat{\phi}_{f,\eta}(x) > t].\]
We also define
\[
\hat{\psi}_{\mcD,\eta}(f)
=
10\Ex_{x \sim \mcD}[\Ind[\hat{\phi}_{f,T,\eta}(x) > 0.6] \cdot ( \phi_{f, \eta}(x) - 0.6)].
\]
\end{definition}}

\subsection{Distances and errors}

\begin{definition}[Error and optimal error]
We will denote the error of a hypothesis on a distribution
\[\err_{\mcD}(h) \coloneqq \Ex_{(x,y) \sim \mcD}
\left[\lfrac 12 |h(x) - y|
\right]\]
and remark that when $h$ is Boolean-valued this is equivalently
\[\err_{\mcD}(h) \coloneqq \Prx_{(x,y) \sim \mcD}[h(x) \ne y].\]

Relative to a sample set we use 
\[\wh{\err}_S(h) = \frac {1}{|S|} \sum_{(x,y) \in S} \lfrac 12 |h(x) - y|.\]
We will often denote by $\opt$ the optimal error of a function $f$ with respect to a concept class $\mcF$:
\[\opt \coloneqq \inf_{g \in \mcF} \err_{\mcD}(g).\]
In this work, $\mcF$ is always the class of linear threshold functions over $\R^d$.
\end{definition}

\subsection{Miscellaneous}
We will use the notation of the form $a=b
\pm c$ as a shorthand for $|a-b|
\leq c$.
We will also drop subscripts, particularly on $\phi$ and $\hat{\phi}$, when they are clear from context.
Throughout this work, wherever there is noise, the noise scale $\eta$ is always $10r$, where $r$ is the desired robustness radius.
We use the notation $t \sim [-1,1]$ to denote that $t$ is drawn uniformly from $[-1,1]$.

\section{Results and pseudocode of our main algorithm }
Our main result is the following. It states the correctness and running time of the main algorithm, {\sc RobustLearn.}

\begin{restatable}[Correctness and complexity of \Cref{alg:RobustLearn}]{theorem}{robustlearn}
\label{thm:main}
Let $\eps \ge d^{-1/7}$, let $r, \delta \in (0,1)$, and let $C$ be a sufficiently large absolute constant. Furthermore, let $\hat{\phi}$ be an $\epsilon$-accurate local noise sensitivity approximator for degree-$k$ PTFs over $\R^d$, where $k \coloneqq \frac{C \log^2 (1/\epsilon)}{\epsilon^2}$, and let the running time and query complexity of $\hat{\phi}$ be $\tau$.

Let $\mathcal{F}$ be the class of linear threshold functions over $\R^d$ and
$\mcD$ be a distribution over $\R^d \times \bits$ whose $\R^d$-marginal is isotropic, subgaussian and log-concave.
The algorithm {\sc RobustLearn}, given
 i.i.d. sample access to the distribution $\mcD$,
with probability at least $1-O(\delta)$ returns a hypothesis $h:\R^d \rightarrow \{\pm 1\}$ for which the following hold:
\[\err_{\mcD}(h) \le \opt + O(\eps)\]
\[\AdvRob_{\mcD,r}(h) \le O(r + \eps).\]
The running time and sample complexity are 
$\poly\left(d^{\log^2(1/\eps)/\eps^2}  \cdot \log(1/\delta) \cdot \tau \right)$.
\end{restatable}
\begin{corollary}
By instantiating $\hat{\phi}$ with the algorithm of \Cref{sec: randomized construction of phi hat}, which evaluates queries in $\poly\left(d^{\log^2(1/\eps)/\eps^2}  \cdot \log(1/\delta) \right)$ time, \textsc{RobustLearn} runs in total time $\poly\left(d^{\log^2(1/\eps)/\eps^2}  \cdot \log(1/\delta) \right)$.
\end{corollary}
\begin{algorithm}
\begin{algorithmic}[1] 
\State \textbf{Input:} sample access to distribution $\mcD$ over $\R^d \times \{\pm 1\}$, error bound $\eps$,\\
\hskip3.5em 
confidence bound $\delta$, robustness radius $r$, \\
\textbf{Uses:}
local noise sensitivity approximator $\hat{\phi}$. \hspace*{\fill} (See \Cref{def: local noise sensitivity approximator})
\State \textbf{Output:} hypothesis $h:\R^d \to \bits$
\State $p:=\textsc{LearnRealValued}(\mcD,
\eps, r)$.
\hspace*{\fill} (See \Cref{alg: LearnRealValued})
\State \Return \textsc{ComputeClassifier}$(\mcD, p,r,\eps,\delta)$.\hspace*{\fill} (See \Cref{alg:ComputeClassifier})
\end{algorithmic}
\caption{$\textsc{RobustLearn}(\mcD, \eps, \delta, r)$:}
\label{alg:RobustLearn}
\end{algorithm}

The algorithm has two ``phases'': \textsc{LearnRealValued}, which solves a convex relaxation of the learning task,
and \textsc{ComputeClassifier}, which rounds the real-valued hypothesis to a Boolean one with small boundary volume.
We present and analyze these algorithms in \Cref{sec: analyzing LearnRealValued} and \Cref{sec:ComputeClassifier} respectively.
\begin{restatable}[Correctness of \Cref{alg: LearnRealValued}, \textsc{LearnRealValued}\label{thm:LearnRealValued}]{theorem}{LearnRealValued}
    Let $r \in (0,1)$, and let $\mcD$ be a distribution over $\R^d\times \{\pm 1\}$ whose $\R^d$-marginal is a subgaussian isotropic log-concave distribution over $\R^d$.
        Let $\mathcal{F}$ be the class of halfspaces on $\R^d$.
        Then, for some sufficiently large absolute constant $C$, the algorithm LearnRealValued, given $\epsilon, \delta$ and sample access to $\mcD$, runs in time $\poly\left({d}^{O(\log^2(1/\epsilon)/\epsilon^2)}\right)$ and will with probability at least $1-O(\delta)$ return a degree-$O(\log^2(1/\epsilon)/\epsilon^2)$ polynomial $p$ for which the following hold:
        \begin{equation}
        \label{eq: p has small average phi hat}
            \Ex_{t \sim [-1,1]} 
            \left[
            \Ex_{(x,y)\sim D}
            \left[\phi_{\sign(p-t),10r}(x)\right]
            \right]
            \leq 100r + O(\epsilon),
        \end{equation}
\begin{equation}
        \label{eq: p has small average psi hat}
\underset{(x,y) \sim \mcD }{\Pr}\left[
\Ex_{t \sim [-1,1]}[\phi_{\sign(p(x)-t),10r}(x)] \ge 0.7\right]
            \leq O(\epsilon),
        \end{equation}
 \begin{equation}
 \label{eq: p has small average error}
            \Ex_{t \sim [-1,1]} 
            \left[\Prx_{(x,y) \sim D}
            \left[\sign(p(x)-t)\neq y\right]\right]
            \leq \opt + O(\epsilon).
        \end{equation}
    \end{restatable}
    
\begin{restatable}[Correctness of \Cref{alg:ComputeClassifier}, \textsc{ComputeClassifier}]{theorem}{ComputeClassifier}
\label{thm:compute-partition} 
Let $\eps \ge d^{-1/7}$ and $\delta, r > 0$.
Let $\hat{\phi}$ be an $\epsilon$-accurate local noise sensitivity approximator for degree-$k$ PTFs over $\R^d$.
 Let $\mcD$ be a distribution over $\R^d \times \bits$ 
such that the $\R^d$-marginal is isotropic and log-concave.
Let $p : \R^d \to \R$ be a polynomial of degree at most $k$ satisfying 

\begin{equation}
\label{eq: premise 1 of LearnClassifier}
\Ex_{t \sim [-1,1]}
\left[
\Ex_{(x,y) \sim \mcD}[\phi_{\sign{(p-t)},10r}(x)]\right] \le 100r + O(\epsilon),\end{equation}
\begin{equation}
\label{eq: premise 2 of LearnClassifier}
\underset{(x,y) \sim \mcD }{\Pr}\left[
\Ex_{t \sim [-1,1]}[\phi_{\sign(p(x)-t),10r}(x)] \ge 0.7\right]
\le O(\epsilon).\end{equation}

Then the algorithm \textsc{ComputeClassifier}$(\mcD, p, r, \eps, \delta)$, given sample access to $\mcD$ and query access to $\hat{\phi}$, outputs a 
hypothesis $h:\R^d \to \bits$ such that the following properties hold with probability at least $1-O(\delta)$:
\[\err_{\mcD}(h) \le \Ex_{t \sim [-1,1]}[\err_{\mcD}(\sign(p-t))]+ O(\eps)\]
\[\AdvRob_{\mcD,r}(h) \le O(r + \eps).\]
The running time and number of queries to $\hat{\phi}$ are $\poly(d^k \cdot 1/\eps \cdot \log(1/\delta))$.
\end{restatable}

\subsection{Proof of \Cref{thm:main}.}
In this section we prove the main theorem assuming \Cref{thm:compute-partition} and \Cref{thm:LearnRealValued}.

\begin{proof}[Proof of \Cref{thm:main}]
By \Cref{thm:LearnRealValued}, with probability $1-O(\delta)$, the output $p$ of {\sc LearnRealValued} is a degree-$k$ polynomial satisfying 
       \begin{equation}
        \label{eq: p has small average phi hat restated}
            \Ex_{t \sim [-1,1]} 
            \left[
            \Ex_{(x,y)\sim D}
            \left[\phi_{\sign(p-t),10r}(x)\right]
            \right]
            \leq 100r + O(\epsilon),
        \end{equation}
        \begin{equation}
        \label{eq: p has small average psi hat restated}
        \underset{(x,y) \sim \mcD }{\Pr}\left[
\Ex_{t \sim [-1,1]}[\phi_{\sign(p(x)-t),10r}(x)] \ge 0.7\right]
            \leq O(\epsilon),
        \end{equation}
 \begin{equation}
 \label{eq: p has small average error restated}
            \Ex_{t \sim [-1,1]} 
            \left[\Pr_{(x,y) \sim D}
            \left[\sign(p(x)-t)\neq y\right]\right]
            \leq \opt + O(\epsilon).
        \end{equation}

This polynomial $p$ is the input to {\sc ComputeClassifier}. 
We see that the premises of \Cref{thm:compute-partition} are satisfied by the guarantees of  \Cref{eq: p has small average phi hat restated} and \Cref{eq: p has small average psi hat restated}.
Taking the conclusions of \Cref{thm:compute-partition} and combining it with
\Cref{eq: p has small average error restated}
we see that the hypothesis $h$ we return satisfies
\[\err_{\mcD}(h) \le \Ex_{t \sim [-1,1]}[\err_{\mcD}(\sign(p-t))]+ O(\eps) \leq \opt +O(\epsilon),\]
\[\AdvRob_{\mcD,r}(h) \le O(r + \eps),\]
which concludes the proof of correctness.

We now analyze the sample complexity and running time. By \Cref{thm:LearnRealValued}, {\sc LearnRealValued} takes $d^{O(\log^2(1/\epsilon)/\epsilon^2)}\poly(\log(1/\delta))$ time and samples. 
By \Cref{thm:compute-partition}, since the degree of $p$ is $O(\log^2(1/\epsilon)/\epsilon^2)$ and $\eps \ge d^{-1/7}$, \textsc{ComputeClassifier} takes $d^{O(\log^2(1/\epsilon)/\epsilon^2)} \poly(\log(1/\delta))$ time and queries to $\hat{\phi}$.
Evaluations of $\hat{\phi}$ take time $\tau$, so the total running time is 
\[\poly(d^{\log^2(1/\epsilon)/\epsilon^2)} \cdot \log(1/\delta) \cdot \tau),\]
and all components succeed with probability $\ge 1 - O(\delta)$.
\end{proof}

\section*{Acknowledgments}
We thank Ronitt Rubinfeld for proving insightful comments on a draft of this paper. We are also thankful to Vinod Vaikuntanathan for helpful discussions regarding backdoors and verifiable robustness.

Jane Lange is supported by NSF Awards CCF-2006664 and CCF-2310818, and NSF Graduate Research Fellowships Program, and Arsen Vasilyan is supported in part by NSF AI Institute for Foundations of Machine
Learning (IFML), NSF awards CCF-2006664,
CCF-2310818 and Big George Fellowship. Part of this work was conducted while the authors were visiting the Simons
Institute for the Theory of Computing.

\bibliographystyle{alpha}
\bibliography{references}

\appendix
\section{Further preliminaries}
\subsection{Randomized implementation of $\hat{\phi}$}
\label{sec: randomized construction of phi hat}
{
Let $T$ be a sufficiently large set of i.i.d. samples from $\mcN(0,I_d)$. Below, we show that the choice of $\hat{\phi}$ as the following empirical estimate will satisfy \Cref{def: local noise sensitivity approximator} with high probability over the choice of $T$:
\[
\tilde{\phi}_{f, T,  \eta}(x) = \frac{1}{|T|}\sum_{z \in T}\frac{|f(x)-f(x+\eta z)|}{2}.
\]
}

\begin{fact}[\cite{Anthony95}]
\label{fact:VC-PTF}
The VC dimension of the class of polynomial threshold functions of degree $k$ is $\Theta(\binom{d}{\le k}) = O(d^k)$.
\end{fact}

\begin{fact}[Concentration of expectations from VC dimension]
\label{fact:VC-convergence}
Let $\mcC$ be a concept class and $\mcD$ be a distribution over $\R^d$. For some sufficiently large absolute constant $C$, the following is true. With probability at least $1-\delta$ over a sample
$S$ of i.i.d. samples from $\mcD$, with $|S| \ge C \cdot VC(\mcC) \log(1/\delta)/\eps^2$, the following holds:
\[\sup_{f \in \mcC}\left | \Ex_{x \sim \mcD}[f(x)] - \frac{1}{|T|}\sum_{x \in T} f(x) \right | \le \eps.\]
\end{fact}

We will call a sample set ``representative'' of distribution $\mcD$ for a concept class if it satisfies the above success condition.

\begin{claim}[Representative samples provide accurate noise sensitivity estimates]
\label{clm:phi-samples-good}
Let $\mcC$ be the class of degree-$k$ PTFs.
Let $T$ be a set of points in $\R^d$ that is representative for $\mcN(0,I_d)$ for $\mcC$ (in the sense of \Cref{fact:VC-convergence}).
Then with probability $\ge 1 -\delta$ over $T$, the following holds for all $f \in \mcC$, $x \in \R^d$, and $\eta \in (0,1)$: 
\[\abs{\tilde{\phi}_{f,T,\eta}(x) - \phi_{f,\eta}(x)} \le \eps.\]
\end{claim}

\begin{proof}
Observe that for any fixed $x$, the function $g(z) = -f(x) \cdot f(x + \eta z)$
is a member of $\mcC$.
Thus by \Cref{fact:VC-convergence}, we have with probability $1-\delta$
that 
\[\abs{\Ex_{z \sim T}[g(z)] - \Ex_{z \sim \mcN(0,I_d)}[g(z)]} \le \eps.\]
The claim follows from observing that $\phi_{f,r}(x) = \frac 12 (\E[g(z)] + 1)$, and thus
\[\abs{\Tilde{\phi}_{f,T,\eta}(x) - \phi_{f,\eta}(x)} = \frac 12 \abs{\Ex_{z \sim T}[g(z)] - \Ex_{z \sim \mcN(0,I_d)}[g(z)]} \le \eps.\]
\end{proof}

\subsection{Distances and errors}

The following facts about the error and local noise sensitivity 
after rounding a real-valued function to Boolean 
follow directly from linearity of expectation.
\begin{fact}[Randomized rounding]
\label{fact:rounding}
Let $f$ and $g$ be real-valued functions.
Let the rounded function $f_t$ be defined as 
\[f_t(x) = \sign(f(x) - t),\]
and likewise for $g_t$.
Then,
\[\Ex_{t \sim [-1,1]}\Brac{\Prx_{(x,y) \sim \mcD}[f_t(x) \ne y]} \leq  \Ex_{(x,y) \sim \mcD}\Brac{\frac{|f(x) - y|}{2}}.\]
\[\Ex_{t \sim [-1,1]}\Brac{\Prx_{(x,y) \sim \mcD}[f_t(x) \ne g_t(x)]} \leq  \Ex_{(x,y) \sim \mcD}\Brac{\frac{|f(x) - g(x)|}{2}}.\]
Similarly,\footnote{See \Cref{sec: randomized construction of phi hat} for the definition of $\tilde{\phi}$.}
\[\Ex_{t \sim [-1,1]}[\phi_{f_t, \eta}(x)] \le \phi_{f,\eta}(x)\quad \text{and} \Ex_{t \sim [-1,1]}[\tilde{\phi}_{f_t, T,\eta}(x)] \le \tilde{\phi}_{f,T,\eta}(x)\]
for any $x$, $\eta$ and $T$.
\end{fact}

Finally, the following is a standard fact, see for example \cite{devroye2018total}.
\begin{fact}[TV distance between one-dimensional Gaussians]
\label{fact: TV distance between Gaussians}
    For all $\mu_1, \mu_2$ in $\R$ and $\sigma_1, \sigma_2$ in $\R_{>0}$ we have
    \[
    \dtv(\mathcal{N}(\mu_1, \sigma_1^2),\mathcal{N}(\mu_2, \sigma_2^2))
    \leq
    \frac{3|\sigma_1^2-\sigma_2^2|}{2\sigma_1^2}+\frac{|\mu_1-\mu_2|}{2\sigma_1}
    \]
\end{fact}

\subsection{Miscellaneous}
We say that $f$ is a \textit{randomized Boolean function} over $\R^d$ if $f$ maps every point in $\R^d$ to a random variable $f(x)$ over $\{0,1\}$. Sometimes we will write $f(x)=\text{Ber}(p)$, to denote that $f(x)=1$ with probability $p$ and $0$ with probability $1-p$. In \Cref{sec: misc uniform convergence}, miscellaneous claims about uniform convergence are shown, which are used in the analysis of \Cref{alg: LearnRealValued}.

\section{Pseudocode and analysis of LearnRealValued}
In this section we present and analyze our algorithm for learning a polynomial with small noise sensitivity and isolation probability.

\begin{algorithm}
\begin{algorithmic}[1]
\caption{$\textsc{LearnRealValued}(D, \eps, r)$:}

\State \textbf{Input:} Sample access to distribution $D$, error bound $\eps$, robustness radius $r$
\State \textbf{Output:} hypothesis $h:\bits^d \to \R$
\State Set $k \coloneqq \frac{C \log^2 (1/\epsilon)}{\epsilon^2}$.
\For{$i$ in $\{1,\cdots, \log(1/\delta)\}$}
\State $S_i \coloneqq \{ Cd^{Ck} \cdot \log(1/\delta)/\eps^2\text{ i.i.d. draws from }D\}$.
\State $T_i \coloneqq \{Cd^{Ck} \cdot \log(1/\delta)/\eps^2\text{ i.i.d. draws from }\mcN(0,I_d)\}$
\State Let $p_i$ be the output of the following convex program:
\begin{itemize}
		\item Domain: polynomials over $\R^d$ of degree $k$.
		\item Minimize: $\sum_{(x,y) \in S} |p(x) - y|$ 
		\item Constraints:
				\begin{itemize}
                \item Define $\Tilde{\phi}(x)_{p,T,10r}:=\frac{1}{|T|} \sum_{z \in T} \frac{\abs{f(x) - f(x + \eta z)}}{2}$
                \item Define $\Tilde{\psi}_{p,T,10r}:=10\cdot\Ind[\Tilde{\phi}(x)_{p,T,10r}>0.6]\cdot(\Tilde{\phi}(x)_{p,T,10r}-0.6)$
						\item Constraint 1: $\frac{1}{|S|} \sum_{x \in S} \Tilde{\phi}_{p,T,10r}(x) \le 100 r + \eps$
						\item Constraint 2:  $\frac{1}{|S|} \sum_{x \in S} \Tilde{\psi}_{p,T,10r}(x) \le \eps$
				\end{itemize}
\end{itemize}
\EndFor
\State $i^*\coloneqq \argmin _i \left( \E_{(x,y) \sim S_i}
            \left[\frac{\abs{p_i(x)-y}}{2}\right] \right)$
\State $p:=p_{i^*}$.
\State \Return $p$
\label{alg: LearnRealValued}
\end{algorithmic}
\end{algorithm}
\subsection{Facts about log-concave distributions.}

\begin{definition}
\label{fact: sub-Gaussian are concentrated}
   A distribution $D$ over $\R^d$ is is called $C$-sub-Gaussian, if for any unit vector $u$: 
\[\Pr_{x\sim D}[|u \cdot x|>t]\leq \exp(-C t^{2})\] for any $t\geq 0$. 
We say $D$ is sub-Gaussian, if $D$ is $C$-sub-Gaussian for some absolute constant $C$. 
\end{definition}

See e.g. \cite{saumard2014log} for the following fact:
\begin{fact}
\label{fact: log-concave under convolution}
    If $D_1$ and $D_2$ are log-concave distributions over $\R^d$, then the convolution of $D_1$ and $D_2$ is also a log-concave distribution over $\R^d$.
\end{fact}

Likewise, the following fact follows directly from the definition of a sub-Gaussian distribution:
\begin{fact}
\label{fact: subGaussian under convolution}
    If $D_1$ and $D_2$ are sub-Gaussian distributions over $\R^d$, then the convolution of $D_1$ and $D_2$ is also a sub-Gaussian distribution over $\R^d$.
\end{fact}

The following fact can be found in \cite{saumard2014log}: 
\begin{fact}
\label{fact: log-concave closed under projections}
   If a distribution $D$ over $\R^d$ is isotropic log-concave, then the projection of $D$ on any linear subspace in $\R^d$ is likewise log-concave. 
\end{fact}

The following two facts can be found in \cite{gollakota2023moment}, \cite{awasthi2017power} and the references therein.

\begin{fact}
\label{fact: log-concave are anticoncentrated}
   If a distribution $D$ over $\R^d$ is isotropic log-concave, then for some absolute constant $C_3$, for any unit vector $u$ and every interval $J \subset \R$, we have $\Pr_{x\sim D}[\langle u, x \rangle \in J]\leq |J|$. 
\end{fact}

\begin{remark}
\label{remark: approximately isotropic is enough}
    It follows from a direct change-of-basis argument that Fact \ref{fact: log-concave are anticoncentrated} still holds  if the premise that the distribution $D$ is isotropic\footnote{Recall that the condition that $D$ is isotropic means that $D$ has zero mean and covariance $I_{d\times d}$.} is replaced by a slightly more general premise that $D$ mean zero and satisfies:
    \[
    I_{d\times d}
    \preceq
    \E_{x\sim D}[xx^T] \preceq 200 I_{d\times d}
    \]
\end{remark}

For the following fact, see e.g. \cite{diakonikolas2018learning} and the references therein:
\begin{fact}
\label{fact: log-concave are tame}
   If a distribution $D$ over $\R^d$ is log-concave, then for some absolute constant $C_4$, if $p$ is a degree-$k$ polynomial for which $\E_{x \sim D}[(p(x))^2]
   \leq 1$, then for any $B>e^k$ it is the case that \[\Pr_{x \sim D}[(p(x))^2>B]
   \leq \exp(-C_4 B^{1/k})\]
\end{fact}
One can use the fact above to conclude the following:
\begin{observation}
\label{obs: under log-concave L1 norm not too small}
   If a distribution $D$ over $\R^d$ is log-concave, then for some absolute constant $C_5$, for every degree-$k$ polynomial $p$ it is the case that 
    \[
    \E_{x \sim D}[|p(x)|] \geq \frac{\sqrt{\E_{x \sim D}[(p(x))^2]}}{(C_5k)^{C_5k}}
    \]
\end{observation}
\begin{proof}
    Without loss of generality, we can assume that  $\E_{x \sim D}[(p(x))^2]= 1$. For any $B\geq 0$, we have
    \begin{equation} 
    \label{eq: lower bound on one norm log-concave}
    \E_{x \sim D}[|p(x)|]
   \geq \frac{ \E_{x \sim D}[|p(x)|^2 \indicator_{|p(x)|\leq B}]
   }{B}
   = \frac{1 - \E_{x \sim D}[|p(x)|^2 \indicator_{|p(x)|>B}]
   }{B}
        \end{equation}
    Using \Cref{fact: log-concave are tame}, we have
    \begin{multline*}
    \E_{x \sim D}[|p(x)|^3]
    \leq
    e^{O(k)}+
    \int_{\beta =0}^{+\infty} \Pr_{x \sim D}[|p(x)|^3>\beta] d\, \beta \leq \\
    e^{O(k)}+
    \int_{\beta' =0}^{+\infty} (\beta')^{3k/2-1} \exp(-C_4 \beta') d\, \beta'
    \leq 
    (O(k))^{O(k)},
    \end{multline*}
which allows us to conclude that 
\[
\E_{x \sim D}[|p(x)|^2 \indicator_{|p(x)|>B}]
\leq
\frac{\E_{x \sim D}[|p(x)|^3]
}{B}
=\frac{(O(k))^{O(k)}
}{B}.
\]
Taking $B$ to equal $(C' k)^{C'k}$ for a sufficiently large absolute constant $C'$, and substituting the bound above into \Cref{eq: lower bound on one norm log-concave} finishes the proof.
\end{proof}

\subsection{Approximating halfspaces with polynomials}


\begin{claim}
\label{claim: halfspaces approximable by polynomials}
    Let $h$ be a linear threshold function over $\R^d$ .  Then, for every absolute constant $C$, there exists a polynomial $p$ of degree $O(1/\epsilon^{2} \log^2(1/\epsilon))$, such that for any $C$-sub-Gaussian log-concave distribution $D$ over $\R^d$ satisfying 
    $
    I_{d\times d}
    \preceq
    \E_{x\sim D}[xx^T] \preceq 200 I_{d\times d}
    $ we have 
    \[
    \E_{x \sim D}[|h(x)-p(x)|] \leq O(\epsilon)
    \]
\end{claim}
\begin{remark}
    Recall that a distribution $D$ is log-concave if the logarithm of the probability density function of $D$ is concave. Note that the constants hidden by the $O(\cdot)$ notation in \Cref{claim: halfspaces approximable by polynomials} are allowed to depend on the sub-Gaussianity constant $C$.
\end{remark}

To prove the claim above, we will need the following proposition from \cite{diakonikolas_bounded_2009}:
\begin{fact}[\cite{diakonikolas_bounded_2009}]
\label{fact: DGJSV polynomial}
There exist absolute constants $C_4$ and $C_5$, such that for every $\epsilon \in (0, 0.1)$, there exists a univariate polynomial $P$ satisfying the following. Denoting $a:=\frac{\epsilon^2}{C_4 \log(1/\epsilon)}$ and $K:=C_5\frac{\log \frac{1}{\epsilon}}{a}=O\left(\frac{\log^2 1/\epsilon}{\epsilon^2}\right)$, it is the case that
    \begin{itemize}
        \item The degree of $P$ is at most $K$.
        \item $P(t)\in[\sign(t), \sign(t)+\epsilon]$ for all $t$ in $[-1/2, -2a]\cup [0,1/2]$.
        \item $P(t) \in [-1, 1+\epsilon]$ for $t$ in $[-2a,0]$
        \item $|P(t)|\leq 2 \cdot (4t)^K$ for all $|t|\geq 1/2$.
    \end{itemize}
\end{fact}

The proof below follows closely the line of reasoning in \cite{diakonikolas_bounded_2009} and is included for completeness. 
\begin{proof}[Proof of \Cref{claim: halfspaces approximable by polynomials}]
For a unit vector $u$
$h(x)=\sign(u\cdot x - \tau)$ (assuming $\tau=0$ for now)

If $|\tau| > \log (1/\epsilon)$ the statement follows immediately by taking $P$ to be the constant polynomial that equals to $\sign(\tau)$ everywhere in $\R^d$. The error between $h$ and $P$ will be upper-bounded by $O(\epsilon)$ 
via \Cref{fact: sub-Gaussian are concentrated} and \Cref{remark: approximately isotropic is enough}.

Otherwise, we have $|\tau| \leq \log (1/\epsilon)$. Let $P$ be the polynomial in \Cref{fact: DGJSV polynomial} and take $w:=C_6 \left(\frac{\log^2 1/\epsilon}{\epsilon^2}\right)^{1/2}$ for a sufficiently large absolute constant $C_6$. Denoting by  $D_{\text{projected}}$ the projection of distribution $D$ on the unit vector $u$, we can bound the error as follows using \Cref{fact: DGJSV polynomial}:
\begin{multline}
\E_{x \sim D}
[
\abs{
\sign(u\cdot x-\tau) - P( (u\cdot x-\tau) /w)
}
]
=\E_{z \sim D_{\text{projected}}}
[
\abs{
\sign(z-\tau) - P((z-\tau)/w)
}
]
\leq\\
\epsilon+
2 \Pr_{z \sim D_{\text{projected}}}
[
-2a \leq z/w-\tau \leq 0]
+ \E_{z \sim D_{\text{projected}}} \left[\left(2 \cdot \left(4\cdot \frac{z-\tau}{w}\right)^K+1\right) \indicator_{|(z-\tau)/w|\geq 1/2}\right]
\end{multline}
The second term above is bounded by $O(aw)$ due to \Cref{fact: log-concave are anticoncentrated} and \Cref{remark: approximately isotropic is enough}. We also note that since $|\tau| \leq \log (1/\epsilon)$ and $w:=C_6 \left(\frac{\log^2 1/\epsilon}{\epsilon^2}\right)^{1/2}$, for sufficiently large absolute constant $C_6$ we gave $|
\tau|/w\leq 0.1
$. Therefore, whenever $(z-\tau)/w|\geq 1/2$ we also have $z\geq 0.4 w$. This allows us to upper bound
\begin{multline}
\label{eq: breaking interval into three regions}
\E_{x \sim D}
[
\abs{
\sign(u\cdot x-\tau) - P( (u\cdot x-\tau) /w)
}
]
\leq\\
\epsilon+O(aw)
+ \E_{z \sim D_{\text{projected}}} \left[\left(2 \cdot \left(4\cdot \frac{z}{w}+0.1\right)^K+1\right) \indicator_{|z|\geq 0.4 w}\right] \leq \\
\epsilon+O(aw)
+ \E_{z \sim D_{\text{projected}}} \left[\left(3 \cdot \left(5\cdot \frac{z}{w}\right)^K\right) \indicator_{|z|\geq 0.4 w}\right]\end{multline}
Decomposing the interval $[0.4, \infty]$ into a union $\bigcup_{i=0}^{\infty}[0.4\cdot 2^i,0.4\cdot 2^{i+1}]$ and using \Cref{fact: sub-Gaussian are concentrated} and \Cref{remark: approximately isotropic is enough} we can now upper-bound the last term above as follows:
\begin{multline}
\label{eq: bound on tail contributions}
\E_{z \sim D_{\text{projected}}} \left[\left(3 \cdot \left(5\cdot \frac{z}{w}\right)^K\right) \indicator_{|z|\geq 0.4 w}\right]  \leq \\
2\sum_{i=0}^{\infty} \left( 3\cdot (4 \cdot 2^i)^K
\Pr_{z \sim D_{\text{projected}}} \left[|z| \geq 2^{i-1} \cdot 0.4w \right] 
\right) \leq \\
6\sum_{i=0}^{\infty}\left( (4 \cdot 2^i)^K \exp(-C (2^{i-3} w)^{2}) \right)
\end{multline}
Substituting $K=O\left(\frac{\log^2 1/\epsilon}{\epsilon^2}\right)$ and $w=C_6 \left(\frac{\log^2 1/\epsilon}{\epsilon^2}\right)^{1/2}$, we can bound each term in the sum above 
\[
(4 \cdot 2^i)^K \exp(-C (2^{i-3} w)^{2})
\leq \exp\left((i+2)\cdot O\left(\frac{\log^2 1/\epsilon}{\epsilon^2}\right)-C\cdot C_6 2^{i-3} \frac{\log^2 1/\epsilon}{\epsilon^2} \right)
\]
For a sufficiently large absolute constant $C_6$ the above is at most $\epsilon/2^{i+1}$, which substituted into \Cref{eq: bound on tail contributions} gives us 
\[
\E_{z \sim D_{\text{projected}}} \left[\left(3 \cdot \left(5\cdot \frac{z}{w}\right)^K\right) \indicator_{|z|\geq 0.4 w}\right] \leq \epsilon.
\]
Finally, substituting the inequality above  into \Cref{eq: breaking interval into three regions}, we get
\[
\E_{x \sim D}
[
\abs{
\sign(u\cdot x) - P( u\cdot x /w)
}
]=O(\epsilon),
\]
finishing the proof.
\end{proof}

    \subsection{Analyzing LearnRealValued.}
    \label{sec: analyzing LearnRealValued}
    
    \LearnRealValued*
We first argue that the run-time is indeed $d^{O(\log^2(1\eps)/\eps^2}$. The algorithm operates with polynomials of degree $k=O(\log^2(1\eps)/\eps^2$ over $\R^d$, which are described via $d^O(k)$ coefficients. It remains to show that the constraints for the optimization task are indeed convex in the coefficients in the polynomial $p$. Recall that the constraints are:
\begin{itemize}
\item Constraint 1: $\frac{1}{|S|} \sum_{x \in S} \Tilde{\phi}_{p,T,10r}(x) \le 100 r + \eps$
						\item Constraint 2:  $\frac{1}{|S|} \sum_{x \in S} \Tilde{\psi}_{p,T,10r}(x) \le \eps$
\end{itemize}
To see that these constraints are indeed convex, first note that, by definition of $\Tilde{\phi}_{p,T,10r}(x)$ and the triangle inequality, we have for any pair of polynomials $p_1, p_2$, $x$ in $R^d$ and $\alpha \in [0,1]$ 
\begin{equation}
\label{eq: convexity of phi}
\Tilde{\phi}_{(\alpha p_1+(1-\alpha)p_2),T,10r}(x)
\leq
\alpha\Tilde{\phi}_{ p_1,T,10r}(x)
+(1-\alpha)\Tilde{\phi}_{p_2,T,10r}(x),
\end{equation}
which implies that Constraint 1 is convex. Likewise, reclalling the definition of $\psi$ 
\[\Tilde{\psi}_{p,T,10r}:=10\cdot\Ind[\Tilde{\phi}(x)_{p,T,10r}>0.6]\cdot(\Tilde{\phi}(x)_{p,T,10r}-0.6),\]
and combining it with \Cref{eq: convexity of phi} we see that
\[
\Tilde{\psi}_{\alpha p_1+(1-\alpha)p_2,T,10r}(x)
\leq
\alpha \Tilde{\psi}_{p_1,T,10r}(x)
(1-\alpha)\Tilde{\psi}_{ p_2,T,10r}(x).
\]

We now proceed to the proof of correctness, the following claim will be key to our argument, which will be proven in the next subsection (\Cref{sec: proving claim for learn real valued}):
\begin{claim}
\label{claim: augmented regression gives good polynomial}
    For any specific iteration $i$ of the algorithm, the polynomial $p_i$ will satisfy the following:
    \begin{equation}
        \label{eq: p has small average phi hat new}
            \E_{(x,y)\sim S_i}
            \left[\Tilde{\phi}_{p_i,T_i,10r}(x)\right]
            \leq 100r + O(\epsilon),
        \end{equation}
        \begin{equation}
        \label{eq: p has small average psi hat new}
            \E_{(x,y)\sim S_i}
            \left[\Tilde{\psi}_{p_i,T_i,10r}(x)\right]
            \leq O(\epsilon),
        \end{equation}
and the following holds with probability at least $0.9$:
 \begin{equation}
 \label{eq: p has small average error new}
            \E_{(x,y) \sim S_i}
    \left[\frac{\abs{p_i(x)-y}}{2}\right]
            \leq \min_{f \in \mathcal{F}}\left( \Pr_{(x,y) \sim D}
            \left[f(x)\neq y\right]\right) + O(\epsilon),
        \end{equation}
\end{claim}

We now prove \Cref{thm:LearnRealValued} assuming \Cref{claim: augmented regression gives good polynomial}.
\begin{proof}[Proof of \Cref{thm:LearnRealValued}]
First we note that for sufficiently large value of the absolute constant $C$, with probability at least $1-O(\delta)$ for all $i$ in $\{1,\cdots, \log 1/\delta\}$, the set $S_i$ will satisfy \Cref{claim: uniform convergence of phi}, \Cref{claim: uniform convergence of isolation probability} and \Cref{fact:VC-convergence for errors}, and the set $T_i$ will satisfy \Cref{claim: uniform convergence of phi}. Furthermore,
since the main iteration is repeated $O(\log 1/\delta)$ times, the we see that with probability $1-O(\delta)$ the polynomial $p=p_{i^{*}}$ will satisfy Equations \ref{eq: p has small average phi hat new}, \ref{eq: p has small average psi hat new} and \ref{eq: p has small average error new}. Denoting $S=S_{i^{*}}$ and $T=T_{i^{*}}$, we have
\begin{multline*}
 \E_{t \sim [-1,1]} 
            \left[
            \E_{(x,y)\sim D}
            \left[\phi_{\sign(p-t),10r}(x)\right]
            \right]
\overbrace{
\leq
 \E_{t \sim [-1,1]} 
            \left[
            \E_{(x,y)\sim S}
            \left[\phi_{\sign(p-t),10r}(x)\right]
            \right] + O(\epsilon)}^{\text{Since $S$ satisfies \Cref{claim: uniform convergence of phi}}}\leq\\
            \underbrace{\leq
            \E_{t \sim [-1,1]} 
            \left[
            \E_{(x,y)\sim S}
            \left[\tilde{\phi}_{\sign(p-t),T,10r}(x)\right]
            \right] + O(\epsilon)}_{\text{Since $T$ satisfies \Cref{clm:phi-samples-good}, and WLOG $\epsilon<0.02$}.}
            \underbrace{
            \leq  \E_{(x,y)\sim S}
            \left[\Tilde{\phi}_{p,T,10r}(x)\right] +O(\epsilon)}_{\text{By \Cref{fact:rounding}}}
            \underbrace{
            \leq 100r + O(\epsilon)}_{\text{By \Cref{eq: p has small average phi hat new}}},
\end{multline*}
which gives us \Cref{eq: p has small average phi hat}. We now derive \Cref{eq: p has small average psi hat} as follows: 


\begin{multline*}
 \underset{(x,y) \sim \mcD }{\Pr}\left[
\E_{t \in [-1,1]}[\phi_{\sign(p(x)-t),10r}] \ge 0.7\right]
\overbrace{
\leq \underset{(x,y) \sim S}{\Pr}[\E_{t \in [-1,1]}[\phi_{\sign(p-t),10r}(x)] \ge 0.67] +O(\epsilon)}^{\text{Since $S$ satisfies \Cref{claim: uniform convergence of isolation probability}}}\leq\\
        \underbrace{\leq
\underset{(x,y) \sim S}{\Pr}[\E_{t \sim [-1,1]}[\tilde{\phi}_{\sign(p-t),T,10r}(x)] \ge 0.65] + O(\epsilon)}_{\text{Since $T$ satisfies \Cref{clm:phi-samples-good}, and WLOG $\epsilon<0.02$}.}
            \underbrace{
            \leq  
    O(1)\E_{(x,y)\sim S}
    \left[\Tilde{\psi}_{p,T,10r}(x)\right] +O(\epsilon).}_{\text{By \Cref{fact:rounding} and definition of $
    \tilde{\psi}$ (\Cref{sec: definitions of phi and psi})}}
            \underbrace{
            \leq  O(\epsilon)}_{\text{By \Cref{eq: p has small average psi hat new}}}
\end{multline*}
Finally, we derive \Cref{eq: p has small average error} as follows:
\begin{multline*}
 \E_{t \sim [-1,1]} 
            \left[
            \err_{D}
    \left[\sign(p-t)\right]
            \right]
\overbrace{
\leq
 \E_{t \sim [-1,1]} 
            \left[
            \wh{\err}_{S}
    \left[\sign(p-t)\right]
            \right] + O(\epsilon)}^{\text{Since $S$ satisfies \Cref{fact:VC-convergence for errors}}}\leq\\
            \underbrace{
            \leq  \E_{(x,y) \sim S}
    \left[\frac{\abs{p(x)-y}}{2}\right] +O(\epsilon)}_{\text{By \Cref{fact:rounding}}}
            \underbrace{
            \leq \opt + O(\epsilon)}_{\text{By \Cref{eq: p has small average error new} and defn of $\opt$}}.
\end{multline*}

\end{proof}

\subsection{Proving \Cref{claim: augmented regression gives good polynomial}}
\label{sec: proving claim for learn real valued}
Finally, we finish this section by proving \Cref{claim: augmented regression gives good polynomial}.
Inspecting the algorithm LearnRealValued, we see that the linear program is always feasible for any input dataset $S_i$, because the all-zeros polynomial will satisfy the constraints in the linear program (i.e. Equations \ref{eq: p has small average phi hat new} and \ref{eq: p has small average psi hat new}). It remains to show that with probability at least $0.9$ there exists a polynomial $p_i$ satisfying all three of Equations \ref{eq: p has small average phi hat new}, \ref{eq: p has small average psi hat new} and 
\ref{eq: p has small average error new}.
    
Suppose $f^*$ is the optimal halfspace for $D$, i.e.
\begin{equation}
\Pr_{(x,y) \sim D}
            \left[f^*(x)\neq y\right]
            =
    \min_{f \in \mathcal{F}}\left( \Pr_{(x,y) \sim D}
            \left[f(x)\neq y\right]\right).
\end{equation}
Let $D_{\text{marginal}}$ denote the $\R^d$-marginal of $D$, and let $D_{\text{marginal}}'$ denote distribution of $x+10r z$, where $x$ is sampled from $D_{\text{marginal}}$ and $z$ is sampled from $\mcN(0,I_d)$.
Since $D_{\text{marginal}}$ is assumed to be sub-Gaussian log-concave, 
\Cref{fact: log-concave under convolution} and \Cref{fact: subGaussian under convolution} tell us that $D_{\text{marginal}}'$ is also sub-Gaussian log-concave. Furthermore, since $D_{\text{marginal}}$ is isotropic (i.e. has mean zero and identity convariance) and $r$ is in $(0,1)$ we see that 
\[
I_{d}
\preceq
\E_{x \sim D_{\text{marginal}}'}[xx^T]
\preceq
101 I_{d}.
\]
Therefore, \Cref{claim: halfspaces approximable by polynomials} tells us that for a sufficiently large absolute constant $C'$ there is a polynomial of $p^*$ of degree $\frac{C' \log^2 (1/\epsilon)}{\epsilon^2}$ for which
\begin{equation}
\label{eq: p star good approximator 1}
    \E_{x \sim D_{\text{marginal}}}\left[\abs{f^*(x)-p^*(x)}\right]
    \leq O(\epsilon)
\end{equation}
\begin{equation}
\label{eq: p star good approximator 2}
    \E_{x \sim D_{\text{marginal}}'}\left[\abs{f^*(x)-p^*(x)}\right]
    \leq O(\epsilon)
\end{equation}
Recall that for each element $(x,y)$ in $S_i$, the variable $x$ is distributed according to 
$D_{\text{marginal}}$. Similarly, for each $(x,y)$ in $S_i$ and $z$ in $T_i$ the sum $x+10rz$ is distributed $D_{\text{marginal}}'$. This allows us to use inequalities above to conclude that with probability at least $0.99$ we have 
\begin{equation}
\label{eq: p star good approximator 3}
            \E_{(x,y)\sim S_i}\left[ |p^*(x)-f^*(x)|\right]
            \leq O(\epsilon)
\end{equation}
\begin{equation}
\label{eq: p star good approximator 4}
\E_{(x,y)\sim S_i, z \sim T_i}\left[ |p^*(x+10rz)-f^*(x+10rz)|\right]
\leq O(\epsilon)
\end{equation}
In the remainder of this section we show that with probability at least $0.9$ over the choice of $S_i$ and $T_i$ the choice $p_i=p^*$ will indeed satisfy Equations \ref{eq: p has small average phi hat new}, \ref{eq: p has small average psi hat new} and 
\ref{eq: p has small average error new}.
\subsubsection{Bounding $\E_{x\sim D_{\text{marginal}}}[\phi_{f^*,10r}(x)]$ and $\E_{x\sim D_{\text{marginal}}}[\psi_{f^*}(x)]$}
Recall that $f^*(x)=\sign(u \cdot x - \tau)$ for some unit vector $u$ and real $\tau$.
 Here, we show that $f^*$ satisfies the following two inequalities
\begin{equation}
\label{eq: for halfspace phi is small}
\E_{x\sim D_{\text{marginal}}}[\phi_{f^*,10r}(x)]\leq 60 r,\end{equation} 
\begin{equation}\E_{x\sim D_{\text{marginal}}}
\label{eq: for halfspace psi is small}
[\psi_{f^*}(x)] = 0. 
\end{equation}
Indeed, fix a specific value of $x$. Recall that $\phi_{f^*,10r}(x)$ is the probability that $\sign(u \cdot x - \tau) \neq \sign(u \cdot (x+10rz) - \tau)$ where $z$ is sampled from $\mcN(0,I_d)$. With probability at least $0.5$ over the choice of $Z$ we have $\sign(u \cdot z)=\sign(u \cdot x - \tau)$ which leands to $\sign(u \cdot x - \tau) \neq \sign(u \cdot (x+10rz) - \tau)$. Thus, $\phi_{f^*,10r}(x)\leq 0.5$ and substituting this into the definition of $\psi$ we see that $\psi(x)=0$, which proves \Cref{eq: for halfspace psi is small}. 

 Now consider points $x$ in $\R^d$ such that $\abs{u \cdot x - \tau}$ is in the interval $[10r\cdot (i-1), 10r \cdot i]$ for an integer $i$. We have the following two observations:
 \begin{itemize}
     \item 
 By \Cref{fact: log-concave are anticoncentrated}, we see that the probability over $x$ sampled from $D_{\text{marginal}}$ that  $\abs{u \cdot x - \tau}$ is in the interval $[10r\cdot (i-1), 10r \cdot i]$ can be upper-bounded by $20r$.
 \item For a fixed $x$ for which $\abs{u \cdot x - \tau}$ is in the interval $[10r\cdot (i-1), 10r \cdot i]$, via Chebyshev's inequality, it follows that 
 \[
 \Pr_{z \sim \mcN(0,I_d)}[\sign(u \cdot x - \tau) \neq \sign(u \cdot (x+10rz) - \tau)]
 \leq
 \max\left(
 1,
 \left(
 \frac{10r}{10 r (i-1)} \right)^2
 \right)
 =\max\left(
 1,
 \left(
 \frac{1}{i-1} \right)^2
 \right) 
 \]
 \end{itemize}
 Combining the two observations above, we see that 
 \[
 \E_{x\sim D_{\text{marginal}}}[\phi_{f^*,10r}(x)]
 \leq \sum_{i=0}^{\infty} 20r \max\left(
 1,
 \left(
 \frac{1}{i-1} \right)^2
 \right)  \leq 60r,
 \]
 which yields \Cref{eq: for halfspace phi is small}.

            
\subsubsection{Bounding $\E_{(x,y)\sim S_i}[\Tilde{\phi}_{p^*,T_i,10r}(x)]$  and  $\E_{(x,y)\sim S_i}[\Tilde{\psi}_{p^*,T_i,10r}(x)]$.}

We observe that 
            \begin{multline}
            \label{eq: phi hat of a+b}
            \E_{(x,y)\sim S_i}[\Tilde{\phi}_{p^*, T_i}(x)]
            = \E_{
            \substack{
            (x,y)\sim S_i\\
            z \sim T_i
            }}\left[\abs{\frac{p^*(x)-p^*(x+10rz)}{2}}\right]
            \leq\\
            \E_{
            \substack{
            (x,y)\sim S_i\\
            z \sim T_i
            }}\left[\abs{\frac{f^*(x)-f^*(x+10rz)}{2}}\right]
            +\E_{
            \substack{
            (x,y)\sim S_i
            }}\left[\abs{\frac{p^*(x)-f^*(x)}{2}}\right]
            +\E_{
            \substack{
            (x,y)\sim S_i\\
            z \sim T_i
            }}\left[\abs{\frac{p^*(x+10rz)-f^*(x+10rz)}{2}}\right]
            \leq \\ \E_{(x,y)\sim S_i}[\Tilde{\phi}_{f^*, T_i,10r}(x)] + 
            \underbrace{
            \E_{(x,y)\sim S_i} |p^*(x)-f^*(x)| + \E_{(x,y)\sim S_i, z \sim T_i} |p^*(x+10rz)-f^*(x+10rz)|
            }_{=O(\epsilon)\text{ by Equations \ref{eq: p star good approximator 3} and \ref{eq: p star good approximator 4}}}.
            \end{multline}
Analogously, by inspecting the definition of $\psi$, we see that for any $x$ we have
\[
\abs{
\Tilde{\psi}_{p^*, T_i}(x)
-
\Tilde{\psi}_{f^*, T_i}(x)
}
\leq
\abs{
f^{*}(x)-p^{*}(x)
}
+
O(1)\cdot
\E_{z \sim T_i}
\left[
\abs{
f^{*}(x+10z)
-
p^{*}(x+10z)
}
\right]
]
\]
Averaging over $x$ in $S_i$, we have \begin{multline}
           \label{eq: psi hat of a+b}
           \E_{(x,y)\sim S_i}[\Tilde{\psi}_{p^*, T_i}(x)]
            \leq \\ \E_{(x,y)\sim S_i}[\Tilde{\phi}_{f^*, T_i,10r}(x)] + O(1)\cdot
            \underbrace{\left(
            \E_{(x,y)\sim S_i} |p^*(x)-f^*(x)| + \E_{(x,y)\sim S_i, z \sim T_i} |p^*(x+10rz)-f^*(x+10rz)|
            \right)}_{=O(\epsilon)\text{ by Equations \ref{eq: p star good approximator 3} and \ref{eq: p star good approximator 4}}}.
            \end{multline}

Recall that $f^*$ is a $\{\pm 1\}$-valued function, which implies that $\Tilde{\phi}_{f^*, T_i,10r}(x+10rz)$ is in $[0,1]$ for all $x$ and $z$. This allows us to use the Chebyshev's inequality to conclude that:
\[
\E_{T_i}
\left[
\abs{
\Tilde{\phi}_{f^*, T_i,10r}(x)
-
\phi_{f^*,10r}(x)
}
\right]
\leq O\left( \frac{1}{\sqrt{|T_i|}}\right).
\]
Inspecting the definition of $\Tilde{\psi}$, we see that for any $x$ it is the case that 
\[
\abs{
\Tilde{\psi}_{f^*, T_i}(x)
-
\psi_{f^*}(x)
}
\leq
O(1)\cdot
\abs{
\Tilde{\phi}_{f^*, T_i,10r}(x)
-
\phi_{f^*,10r}(x),
}
\]
which implies that
\[
\E_{T_i}
\left[
\abs{
\Tilde{\psi}_{f^*, T_i}(x)
-
\psi_{f^*}(x)
}
\right]
\leq O\left( \frac{1}{\sqrt{|T_i|}}\right)
\]
Averaging the inequalities above over $x$ in $S_i$, we have:
\[
\E_{T_i}
\left[
\E_{(x,y) \sim S_i}
\left[
\abs{
\Tilde{\phi}_{f^*, T_i,10r}(x)
-
\phi_{f^*,10r}(x)
}
\right]
\right]
\leq O\left( \frac{1}{\sqrt{|T_i|}}\right)
\]

\[
\E_{T_i}
\left[
\E_{(x,y) \sim S_i}
\left[
\abs{
\Tilde{\psi}_{f^*, T_i}(x)
-
\psi_{f^*}(x)
}
\right]
\right]
\leq O\left( \frac{1}{\sqrt{|T_i|}}\right)
\]
Thus, with probability at least $0.99$ over the choice of $T_i$, we have
\begin{equation}
\label{eq: concentration for phi over T}
\abs{
\E_{(x,y) \sim S_i}
\left[
\Tilde{\phi}_{f^*, T_i,10r}(x)
\right]
-
\E_{(x,y) \sim S_i}
\left[
\phi_{f^*,10r}(x)
\right]
}
\leq
\E_{(x,y) \sim S_i}
\left[
\abs{
\Tilde{\phi}_{f^*, T_i,10r}(x)
-
\phi_{f^*,10r}(x)
}
\right]
\leq O\left( \frac{1}{\sqrt{|T_i|}}\right)
\end{equation}

\begin{equation}
\label{eq: concentration for psi over T}
\abs{
\E_{(x,y) \sim S_i}
\left[
\Tilde{\psi}_{f^*, T_i}(x)
\right]
-
\E_{(x,y) \sim S_i}
\left[
\psi_{f^*}(x)
\right]
}
\leq
\E_{(x,y) \sim S_i}
\left[
\abs{
\Tilde{\psi}_{f^*, T_i}(x)
-
\psi_{f^*}(x)
}
\right]
\leq O\left( \frac{1}{\sqrt{|T_i|}}\right)
\end{equation}

Again, recalling that $f^*$ is a $\{\pm 1\}$-valued function, we see that $\phi_{f^*,10r}(x)$ is in $[0,1]$ for all $x$ and $\psi_{f^*}(x)$ is likewise bounded by $O(1)$ in absolute value. This allows us to use the Chebyshev's inequality to conclude that with probability at least $0.99$ over the choice of $S_i$ we have:
\begin{equation}
\label{eq: concentration for phi over S}
\abs{
\E_{(x,y) \sim S_i}
\left[
\phi_{f^*,10r}(x)
\right]
-
\E_{x \sim D_{\text{marginal}}}
\left[
\phi_{f^*,10r}(x)
\right]
}
\leq O\left( \frac{1}{\sqrt{|S_i|}}\right)
\end{equation}
\begin{equation}
\label{eq: concentration for psi over S}
\abs{
\E_{(x,y) \sim S_i}
\left[
\psi_{f^*}(x)
\right]
-
\E_{x \sim D_{\text{marginal}}}
\left[
\psi_{f^*}(x)
\right]
}
\leq O\left( \frac{1}{\sqrt{|S_i|}}\right)
\end{equation}
Overall, combining \Cref{eq: for halfspace phi is small}, \Cref{eq: concentration for phi over S}, \Cref{eq: concentration for phi over T} and \Cref{eq: phi hat of a+b} we see that with probability at least $0.97$ over the choice of $S_i$ and $T_i$ it is the case that
\[
\E_{(x,y)\sim S_i}[\Tilde{\phi}_{p^*,10r}(x)] \leq 60r + O(\epsilon).
\]
Similarly, combining \Cref{eq: for halfspace psi is small}, \Cref{eq: concentration for psi over S}, \Cref{eq: concentration for psi over T} and \Cref{eq: psi hat of a+b} we see that with probability at least $0.97$ over the choice of $S_i$ and $T_i$  it is the case that
\[
\E_{(x,y)\sim S_i}[\Tilde{\psi}_{p^*,T_i}(x)] \leq  O(\epsilon).
\]

\subsubsection{$p^*$ has a small empirical error.}
Finally, we argue that with probability $0.99$ we have $\E_{(x,y)\sim S_i}\left[
    \frac{|p^*(x)-y|}{2}
    \right] \leq \E_{(x,y)\sim D}\left[
    \frac{|f^*(x)-y|}{2}
    \right] + O(\epsilon)$. Indeed, Hoeffding's inequality implies that with probability at least $0.995$ we have
    \[
    \E_{(x,y)\sim S_i}\left[
    \frac{|f^*(x)-y|}{2}\right]
    \leq 
    \E_{(x,y)\sim D}\left[
    \frac{|f^*(x)-y|}{2}\right]
    +\epsilon,
    \]
    whereas Markov's inequality tells us that with probability 0.995 we have
    \[
    \E_{(x,y)\sim S_i}\left[
    \frac{|f^*(x)-p^*(x)|}{2}\right]
    \leq O(1) \cdot \E_{(x,y)\sim D}\left[
    \frac{|f^*(x)-p^*(x)|}{2}\right]
    =O(\epsilon).   
    \]
    Overall, we see that with probability at least 0.99 it holds that
    \[
    E_{(x,y)\sim S_i}\left[
    \frac{|p^*(x)-y|}{2}
    \right]
    \leq
    E_{(x,y)\sim S_i}\left[
    \frac{|f^*(x)-y|}{2}
    \right]
    +
    E_{(x,y)\sim S_i}\left[
    \frac{|p^*(x)-f^*(x)|}{2}
    \right] \leq \E_{(x,y)\sim D}\left[
    \frac{|f^*(x)-y|}{2}\right]
    +O(\epsilon)
    \]

\section{Pseudocode and analysis of ComputeClassifier}
\subsection{Local correction of adversarial robustness}

\begin{theorem}[Correctness and complexity of \Cref{alg: RobustnessLCA}]
\label{lem:robustness-lca}
Let $\hat{\phi}$ be an $\epsilon$-accurate local noise sensitivity approximator for degree-$k$ PTFs over $\R^d$, and let $\tau$ be the run-time and query complexity of $\hat{\phi}$.
Then there exists an algorithm $\textsc{RobustnessLCA}(x, g, r)$ that takes $x$ in $\R^d$, query access to a function $g:\R^d \to \bits$, perturbation size parameter $r$ in $(0,1)$. The algorithm makes $O(\tau)$ queries to $g$, runs in time $O(\tau+d)$, and satisfies the specifications below.

Let $\mcD$ be a distribution over $\R^d$, and let $g$ be a degree-$k$ PTF. 
Let the function $h: \R^d \rightarrow \bits$ be defined as $h(x):=\textsc{RobustnessLCA}(x, g,  r)$.
Then the following properties hold:
\begin{itemize}
	\item[a)] 
        $\Pr_{x \sim \mcD}[g(x) \ne h(x)] \le 
		\wh{\iso}_{\mcD,10r}(g, 0.8)$.
    \item[b)] 
$\Pr_{x \sim \mcD}[\phi_{g, 10r}(x) \leq 0.1] \ge 1-O\left(\wh{\NS}_{\mcD,10r}(g)\right)$.
    \item[c)]
        For every $x$ in $\R^d$, if $\hat{\phi}_{g, 10r}(x) \leq 0.1$, then for every $x'$ with $\norm{x'-x}\leq r$ we have $h(x')=h(x)$ .
\end{itemize}
\end{theorem}

\begin{algorithm}
\begin{algorithmic}[1]
\caption{$\textsc{RobustnessLCA}(x, g, r)$:}
\State \textbf{Input:} point $x \in \R^d$, black-box representation of a function $g: \R^d \to \bits$, \\
\hskip 3.5em
 robustness radius $r$
\State \textbf{Uses:}
local noise sensitivity approximator $\hat{\phi}$. \hspace*{\fill} (See \Cref{def: local noise sensitivity approximator})
\State \textbf{Output:} $b \in \bits$
\If{$\hat{\phi}_{g, 10r}(x) > 0.8$}
\State \Return $-g(x)$
\Else
\State \Return $g(x)$
\EndIf
\label{alg: RobustnessLCA}
\end{algorithmic}
\end{algorithm}

\begin{proof}
Inspecting the algorithm \textsc{RobustnessLCA}, we see that condition (a) follows directly, as only isolated points have their labels changed. Meanwhile, condition (b) follows directly from Markov's inequality and the definition of $\wh{\NS}$ (\Cref{sec: noise sensitivity approximators}).

To show condition (c), we first recall that by the assumption that $\hat{\phi}$ is an $\eps$-accurate approximator, for every degree-$k$ polynomial threshold function $g$ and every $x$ in $\R^d$ we have
   \begin{equation}
   \label{eq: phihat is close to phi}
    \abs{
    \hat{\phi}_{g, 10r}(x)
    -
    \phi_{g, 10r}(x)
    }
    \leq \epsilon \leq 0.01,
    \end{equation}
where for the last we assumed without loss of generality that $\eps<0.01$.
This implies that whenever $\hat{\phi}_{g, 10r}(x) \leq 0.1$ we have $\phi_{g, 10r}(x) \leq 0.11$. Suppose $x'$ has distance at most $r$ from $x$. Since $\phi_{g,r}(x)$ is defined as 
\[\phi_{g,10r}(x) \coloneqq \Prx_{z \sim N(0, I_d)}[g(x) \ne g(x + 10r z)],\]
we see that 
\begin{multline*}
\Prx_{z \sim N(0, I_d)}[g(x) \ne g(x' + 10r z)]
\leq \\
\Prx_{z \sim N(0, I_d)}[g(x) \ne g(x + 10r z)] +
d_{TV}(N(x,100 r^2 I_d), N((x'),100 r^2 I_d)) \leq \\
\Prx_{z \sim N(0, I_d)}[g(x) \ne g(x + 10r z)]
+ \frac{\norm{x-x'}}{20r},
\end{multline*}
where the last inequality is implied by \Cref{fact: TV distance between Gaussians}. Overall,  since we know that $\phi_{g, 10r}(x) \leq 0.11$ and $\norm{x-x'}\leq r$, we have
\[
\Prx_{z \sim N(0, I_d)}[g(x) \ne g(x' + 10r z)]
\leq 
0.16.
\]

Now, we consider the following two cases:
\begin{itemize}
    \item Suppose $g(x')=g(x)$: then the above implies that $ \phi_{g, 10r}(x') \leq 0.16$, which combined with \Cref{eq: phihat is close to phi} implies that $\hat{\phi}_{g, 10r}(x')\leq 0.17$. Therefore, by inspecting the algorithm  \textsc{RobustnessLCA} we see that $h(x')=g(x')=g(x)=h(x)$. 
    \item On the other hand, suppose  $g(x')=-g(x)$. Then, we have $ \phi_{g, 10r}(x')\geq 1-0.16=0.84$, which combined with \Cref{eq: phihat is close to phi} implies that $\hat{\phi}_{g, 10r}(x')\geq 0.83$. Therefore, by inspecting the algorithm  \textsc{RobustnessLCA} we see that $h(x')=-g(x')=g(x)=h(x)$.
\end{itemize}
In either case, we have $h(x')=h(x)$.

\end{proof}

\subsection{Finding good rounding thresholds}

\begin{restatable}[Correctness of \Cref{alg:ComputeRoundingThresholds}]{theorem}{roundingthresholds}
\label{thm:ComputeRoundingThresholds}
Let $\eps \in (0,1)$ and $\mcD$ be a distribution over $\R^d \times \{\pm 1\}$. Let $\hat{\phi}$ be an $\epsilon$-accurate local noise sensitivity approximator for degree-$k$ PTFs over $\R^d$. 
Let $p : \R^d \to \R$ be a polynomial of degree $\le k$ satisfying the following: 

\begin{equation}
\label{eq: premise on phi 1}
\Ex_{t \sim [-1,1]}
\left[
\Ex_{(x,y) \sim \mcD}[\phi_{\sign{(p-t)},10r}(x)]\right] \le 100r + O(\epsilon),\end{equation}
\begin{equation}
\label{eq: premise on phi 2}
\underset{(x,y) \sim \mcD }{\Pr}\left[
\E_{t \sim [-1,1]}[\phi_{\sign(p(x)-t),10r}(x)] \ge 0.7\right]\le  O(\eps).\end{equation}


The algorithm \textsc{ComputeRoundingThresholds} 
outputs real numbers $t_1, t_2, t_3, t_4$ and $w_1, w_2, w_3, w_4 \in [0,1]$ 
such that $\sum_i w_i=1$.
Let the hypotheses $g_1, g_2, g_3, g_4$ be defined as 
\[g_i(x) = \sign(p(x) - t_i).\]
Then the following properties hold with probability at least $1 - O(\delta)$:
\begin{itemize}
    \item $\sum_i w_i \cdot \err_\mcD(g_i) \le \Ex_{t \sim [-1,1]}[\err_{\mcD}(\sign(p-t))]+ O(\eps)$
	\item $\sum_i w_i \cdot \wh{\NS}_{\mcD,10r}(g_i) \le 200r + O(\eps)$.
	\item $\sum_i w_i \cdot \wh{\iso}_{\mcD,10r}(g_i, 0.8) \le  O(\eps)$.
\end{itemize}
\end{restatable}

\begin{algorithm}
\begin{algorithmic}[1]
\caption{$\textsc{ComputeRoundingThresholds}
(D, p, r, \eps, \delta)$:}
\State \textbf{Input:} sample access to $D$, robustness radius $r$,\\
\hskip3.5em error bound $\eps$, confidence bound $\delta$\\
\textbf{Uses:}
local noise sensitivity approximator $\hat{\phi}$. \hspace*{\fill} (See \Cref{def: local noise sensitivity approximator})
\State \textbf{Output:} thresholds $\vec{t} = t_1, t_2, t_3, t_4 \in [-1,1]$ and weights $\vec{w} = w_1, w_2, w_3, w_4 \in [0,1]$
\State Initialize $Q \coloneqq \varnothing$.
\State $M \gets 100/\eps^3 \cdot \log^2(1/\delta)$ i.i.d. samples $(x,y)$ from $\mcD$
\For{$i \in [\log(1/\delta)]$}
\For{$j \in [100/\eps^2]$}
\State Draw rounding threshold $t$ u.a.r from $[-1,1]$ and define rounded function $p_t(x) \coloneqq \sign(p(x) - t)$.
\State Let $\err_t \coloneqq \err_M(p_t)$.
\State Let $\NS_t \coloneqq \frac{1}{|_M|} \cdot \sum_{x\in M} \hat{\phi}_{p_t, 10r}(x)$.
\State Let $\iso_t \coloneqq \frac{1}{|M|} \cdot \sum_{x \in M} \Ind[\hat{\phi}_{p_t, 10r}(x) \ge 0.8]$.
\State Let $q_t \coloneqq (\err_t, \NS_t, \iso_t)$.
\State Add $q_t$ to $Q$.
\EndFor
\State $\wh{\opt} := \frac{1}{|Q|}\sum_{(\err_t, \NS_t, \iso_t) \in Q} (\err_t)$
\For{tuple $q_1, q_2, q_3, q_4 \in Q^4$}
\State Solve the linear program with variables $w_1, w_2, w_3, w_4$ defined by the following constraints:
\begin{itemize}
		\item $\sum_{i=1}^4 w_i = 1$
		\item $\sum_{i=1}^4 w_i q_{i, 1} \le \wh{\opt} + C\eps$ \hfill{$C$ is a sufficiently large absolute constant.} 
		\item $\sum_{i=1}^4 w_i q_{i, 2} \le 200r+C \eps$ 
        \item $\sum_{i=1}^4 w_i q_{i,3} \le C \eps$
\end{itemize}
\State If a solution is found, \Return $\vec{t} = t_1, t_2, t_3, t_4$ and $\vec{w} = w_1, w_2, w_3, w_4$.
\EndFor
\EndFor
\State \Return $\bot$
\label{alg:ComputeRoundingThresholds}
\end{algorithmic}
\end{algorithm}

We argue that with high probability over the choice of the $O(1/\eps^2)$ random rounding thresholds,
the equal-weighted mixture of the rounded functions has each of the desired properties.
\begin{claim}[Properties of the equal mixture of all rounding thresholds]
\label{clm:rounding_expectation}
Let $\eps \in (0,1)$.
Let $p : \R^d \to \R$ be a polynomial of degree $\le k$ satisfying the following: 
\begin{itemize}
\item $\Ex_{t \sim [-1,1]}
\left[
\Ex_{(x,y) \sim \mcD}[\hat{\phi}_{\sign{(p-t)},10r}(x)]\right] \le \alpha$
  \item $ \underset{(x,y) \sim \mcD }{\Pr}\left[
\E_{t \in [-1,1]}[\hat{\phi}_{\sign(p(x)-t),10r}(x)] \ge 0.75\right]\le  O(\eps)$.

\end{itemize}
Let $\Theta$ be a set of $100/\eps^2$ real numbers drawn uniformly and independently from $[-1,1]$. Then the following hold with probability $\ge 0.99$:
\begin{itemize}
\item[(a)] $\Ex_{t \in \Theta}[\err_\mcD(\sign(p - t))] \le \E_{t \sim [-1,1]}[\err_{\mcD}(\sign(p-t))] + O(\eps)$
\item[(b)] $\Ex_{t \in \Theta}[\Ex_{x \sim \mcD}[\hat{\phi}_{\sign(p-t),10r}(x)]] \le 2\alpha + O(\eps)$
\item[(c)] $\Ex_{t \in \Theta}[\Pr_{x \sim \mcD}[\hat{\phi}_{\sign{(p-t)},10r}(x)>0.8]] \le O(\eps)$.
\end{itemize}
\end{claim}
\begin{proof}[Proof of \Cref{clm:rounding_expectation}]
We will show that each of the conditions holds with high probability, then union bound.
\begin{itemize}
\item[(a)] 
Let $X=\sum_{t \in \Theta}\err_\mcD(\sign(p - t))$.
We have $\Ex_{\Theta}[X] = |\Theta| \cdot \Ex_{t \sim [-1,1]}[\err_{\mcD}(\sign(p-t))]$ and we apply Hoeffding's 
inequality:
\begin{multline*}
\Pr\Brac{\Ex_{t \sim \Theta}[\err_\mcD(\sign(p-t))] > \Ex_{t \sim [-1,1]}[\err_{\mcD}(\sign(p-t))] + \eps} \\
= \Pr\Brac{(X - \E[X])/|\Theta| > \eps} \le \exp(-2\eps^2 \cdot 100/\eps^2) \le 0.001. 
\end{multline*}
\item[(b)] We have by assumption that
\[\Ex_{t \sim [-1,1]} \Ex_{x \sim \mcD}[\hat{\phi}_{\sign(p-t)}(x)] \le \alpha.\] Let $X = \sum_{t \in \Theta} \Ex_{x \sim \mcD}[\hat{\phi}_{\sign(p-t),10r}(x)]$,
which has expectation $\le |\Theta|\cdot \alpha$.
Then by Hoeffding's inequality, we have 
\[\Pr[X/|\Theta| > \alpha+\eps] \le \Pr[(X - \E[X])/|\Theta| > \eps] \le \exp(-2\eps^2 \cdot 100/\eps^2) \le 0.001.\]
\item[(c)] We have by assumption that $\Pr_{x \sim \mcD}[\Ex_{t \sim [-1,1]} [\hat{\phi}_{\sign{(p-t)}]}(x) \ge 0.75] \le O(\epsilon)$. 
We want to bound 
\[\Ex_{t \sim \Theta}\Prx_{x \sim \mcD}[\hat{\phi}_{\sign(p-t)}(x) > 0.75].\]

We can upper bound this by 
\begin{multline*}
O(\eps) + \Ex_{t \sim \Theta}\Prx_{x\sim \mcD}[\hat{\phi}_{\sign(p-t)}(x) - \Ex_{t' \sim [-1,1]}[\hat{\phi}_{\sign(p-t)}(x)] > 0.05] =\\ O(\eps) + \Prx_{x \sim \mcD} \Prx_{t \sim \Theta} [\hat{\phi}_{\sign(p-t)}(x) - \E_{t' \sim [-1,1]}[\hat{\phi}_{\sign(p-t')}(x)] > 0.05].
\end{multline*}
Let $X = \sum_{t \in \Theta}\hat{\phi}_{\sign(p-t)}(x)$; by Hoeffding's inequality we have 
\begin{align*}
\Prx_{\Theta}\Brac{\Ex_{t \sim \Theta}[\hat{\phi}_{\sign(p-t)}(x)] > \E_{t' \sim [-1,1]}[\hat{\phi}_{\sign(p-t')}(x)] + 0.05} &= \Pr[(X - \E[X]/|\Theta|] > 0.05 \\
&\le \exp(-2\cdot 0.025 \cdot (100/\eps^2)) \\
&\le \eps/1000.
\end{align*}


Thus, combining all three probabilities, we have:
\[\Prx_{x \sim \mcD} \Prx_\Theta \Prx_{t \sim \Theta} [\hat{\phi}_{\sign(p-t),10r}(x) -\E_{t' \sim [-1,1]}[\hat{\phi}_{\sign(p-t'),10r}(x)] > 0.05] \le \eps/1000 \]
Then we apply a Markov bound:
\[\Prx_{\Theta}\Brac{\Prx_{x \sim \mcD, t \sim \Theta}[\hat{\phi}_{\sign(p-t),10r}(x) -\E_{t' \sim [-1,1]}[\hat{\phi}_{\sign(p-t'),10r}(x)] > 0.05] > \eps} \le 0.001\]


 \end{itemize}
 Union bounding over all the conditions, we have that all three conditions hold simultaneously with probability
 at least $0.99$ over the choice of $\Theta$.

\end{proof}
Now we will show the existence of a mixture using four rounding thresholds from $\Theta$, and show that \textsc{ComputeRoundingThresholds} finds it.
We will make use of the following well-known theorem of Carath\'eodory.
\begin{theorem}[Carath\'eodory's theorem \cite{Caratheodory1907}]
\label{thm:caratheodory}
Let $P \subset \R^k$ be a set of points and $\mathrm{Conv}(P)$ be its convex hull. For any 
$p \in \mathrm{Conv}(P)$, there exists a set $S\subseteq P$ of $k+1$ points such that $p$ can be written as 
a convex combination of points in $S$.
\end{theorem}

\begin{proof}[Proof of \Cref{thm:ComputeRoundingThresholds}]
We first note that since $\hat{\phi}$ is an $\epsilon$-accurate local noise sensitivity approximator for degree-$k$ PTFs, Equations \ref{eq: premise on phi 1} and \ref{eq: premise on phi 2} respectively imply that\footnote{We assume $\eps\leq 0.05$ without loss of generality.} \begin{equation}
\Ex_{t \sim [-1,1]}
\left[
\Ex_{(x,y) \sim \mcD}[\hat{\phi}_{\sign{(p-t)},10r}(x)]\right] \le 100r + O(\epsilon),\end{equation}
\begin{equation}
\underset{\substack{x \sim \mcD\\ t \sim [-1,1]}}{\Pr}[\hat{\phi}_{\sign(p-t),10r}(x) > 0.75] 
\le O(\epsilon).\end{equation}
Thus, the assumptions of \Cref{clm:rounding_expectation} hold with $\alpha = 100r$.
For a threshold $t$, let the three-dimensional point $q_t \coloneqq (\err_t, \NS_t, \iso_t)$ be defined as in ${\textsc{ComputeRoundingThresholds}}$.
First we argue that each of these estimates is $O(\eps)$-accurate.
Each estimate is the expectation of a random variable bounded in $[0,1]$, so by Hoeffding's inequality we have
\[\Pr[|\err_t - \err_{\mcD}(\sign(p-t))| \ge \eps] \le \exp(-2\eps^2 \cdot |M|),\]
and likewise for each $\NS_t$ and $\iso_t$. 
We set $M$ large enough that this probability is $\le \delta/2 \cdot \exp(-100/\eps)$, giving us more than enough room to union bound over the $\log(1/\delta) \cdot 300/\eps^2$ estimates. Likewise, by the Hoeffding bound, we have with probability $1-\delta$ that
\[
\abs{\wh{\opt}
-
\E_{t \sim [-1,1]}[\err_{\mcD}(\sign(p-t))]
}
\leq \epsilon.
\]
Thus with probability at least $1  - \delta$, all estimates are $\eps$-accurate.
Then, by \Cref{clm:rounding_expectation}, we have that with large constant probability over the randomness of the set $\Theta$, there exists a convex
combination $q^\star$ of the points $\{q_t: t \in \Theta\}$ that satisfies the linear constraints
\begin{itemize}
		\item $q^\star[1] \le \E_{t \sim [-1,1]}[\err_{\mcD}(\sign(p-t))] + O(\eps)
        \leq \opt + C\eps
        $ 
\item $q^\star[2] \le 200r + C\eps$
\item $q^\star[3] \le C\eps$,
\end{itemize}
for a sufficiently large constant $C$.
That linear combination is the equal-weighted mixture $\Ex_{t \in \Theta}[q_t]$.
By \Cref{thm:caratheodory}, since the $q_t$'s are points in $\R^3$, 
it is possible to write $q^\star$ as a convex combination of four members of $\{q_t:t\in \Theta\}$.
Thus, for some 4-tuple in $Q^4$, the linear program is feasible.
\textsc{ComputeRoundingThresholds} searches every possible tuple, guaranteeing that a solution is returned.

By repeating $\log(1/\delta)$ times with independent draws of $\Theta$, the large constant success probability is boosted to $1-\delta$. 
\end{proof}

\newcommand{\trunc}{\mathrm{trunc}}
\subsection{Randomized set partitioning}
\label{sec:ComputeClassifier}

\begin{algorithm}
\begin{algorithmic}[1]
\caption{$\textsc{ComputeClassifier}(\mcD, p, r, \eps, \delta)$:}
\State \textbf{Input:} sample access to $\mcD$, polynomial $p$,\\
\hskip3.5em robustness radius $r$, error bound $\eps$, confidence parameter $\delta$   \\
\textbf{Uses:}
local noise sensitivity approximator $\hat{\phi}$. \hspace*{\fill} (See \Cref{def: local noise sensitivity approximator})
\State \textbf{Output:} Classifier $h:\R^d \rightarrow \{\pm 1\}$. 
\State $(t_1, t_2, t_3,t_4, w_1, w_2, w_3, w_4)\leftarrow$ $\textsc{ComputeRoundingThresholds}
 (\mcD, p, r, \eps, \delta)$. 
 \Statex \hspace*{\fill} (See \Cref{alg:ComputeRoundingThresholds}).
\State $\mathrm{S_{\mathrm{test}}} \gets C \log^2(1/\delta)/\eps^2$ samples from $\mcD$ (for sufficiently large constant $C$)
\For{$i \in [\log(1/\delta)]$}
\State $M \gets 10 d^3$ i.i.d. samples $(x,y)$ from $\mcD$
\State For $i\in\{1,...,4\}$, define the following Boolean functions: let 
\begin{align*}
&h_i(x)=\textsc{RobustnessLCA}(x, \sign(p - t_i), r) \tag{See \Cref{alg: RobustnessLCA}} \\
&\mathrm{RobustIndicator}_i(x) \coloneqq \Ind[\hat{\phi}_{\sign(p - t_i),10r}(x) \le 0.1]
\end{align*}
\State For each $i \in \{1,...,4\}$, let
\[\hat{\mu}_i \leftarrow \Ex_{(x,y) \in M} \Ind[y\neq h_i(x)]\]
\[\hat{\mu}_i' \leftarrow \Ex_{(x,y) \in M} \mathrm{RobustIndicator}_i(x).\]
\State Obtain an orthonormal collection of vectors $\{u_1,\cdots, u_{d/2}\}$ orthogonal to $\text{span}(\mu_1, \cdots, \mu_{4}, \mu_{1}',\cdots, \mu_{4}')$.
\State Generate a uniformly random unit vector $u^\star$ in $\text{span}(u_1,\cdots, u_{d/2})$. 
\State Compute a partition of the real line into intervals $J_1,\ldots J_4$ of Gaussian mass $w_1\ldots,w_4$ respectively.
\State Check if the following hold for each $i \in [4]$:
\[\Ex_{(x,y) \sim S_{\test}}[\Ind[y\neq h_i(x)] \land \langle x, u^\star \rangle \in J_i] = \Ex_{(x,y) \sim S_{\test}}[\Ind[y\neq h_i(x)] \cdot w_i \pm C' \eps\]
 \Statex \hspace*{\fill} (for sufficiently large constant $C'$).
\[\Ex_{(x,y) \sim S_{\test}}[\RobInd_i(x) \land \langle x, u^\star \rangle \in J_i] = \Ex_{x \sim S_{\test}}[\RobInd_i(x)] \cdot w_i \pm C' \eps)\]
\State If so, \Return the function $h$ defined

\[h(x) = \sum_{i=1}^4 h_i(x) \cdot \Ind[\langle u^\star, x\rangle \in J_i]\]
\EndFor
\State \Return $\bot$
\label{alg:ComputeClassifier}
\end{algorithmic}
\end{algorithm}

In this section, we will prove the correctness of {\sc ComputeClassifier}.
As explained in the previous section,
the output of {\sc ComputeRandomThresholds} is a list of four weights and four thresholds for rounding.
The goal is for a small collection of sets --- the sets of points misclassified by each threshold and the
sets of points for which \textsc{RobustnessLCA} guarantees robustness --- to each be partitioned such that $w_i$ of their mass falls into the $i^{th}$ part for each $i$.

The partition is a set of intervals along a unit vector chosen uniformly from the subspace orthogonal to the
estimated mean vectors of the sets we want to partition.
The intervals are chosen to have mass $w_1,\ldots, w_4$ under the one-dimensional standard Gaussian.
We will argue that with high probability over the choice of the random unit vector, all of our sets will be
approximately normally distributed in their projections on this vector, and thus the intervals will contain
the right proportion of their mass.
The key fact underlying this argument comes from the following theorem of \cite{dasgupta2006concentration}:

\begin{theorem}[Concentration of random projections: special case of Theorem 11 of \cite{dasgupta2006concentration}]
		\label{thm:projection-concentration}
Let $D$ be a distribution over $\R^d$ with mean zero and covariance bounded by $\lambda_{\max}$ in operator
norm. Assume also that $D$ is supported only on points $x$ satisfying $\norm{x}_2 \ge \nu \sqrt{d}$. For a unit
vector $v$, let $D_v$ denote the distribution of $\langle x, v\rangle, x \sim D$, and let $D_{\norm{\cdot}}$ 
denote the distribution of $\norm{x}_2/\sqrt{d}, x \sim D$. With probability at least 
\[1 - \frac{\lambda_{\max} \ln(1/\eps)}{\eps^3 \nu^2} \cdot \exp\Paren{-\Omega\Paren{\frac{\eps^4 \nu^2 d}{\lambda_{\max}\ln(1/\eps)}}}\]
over a uniformly random unit vector $v$, the following holds for all intervals $J$:
\[\Abs{\Prx_{z \sim D_v}[z \in J] - \Prx_{\substack{\sigma \sim D_{\norm{\cdot}} \\ z \sim \mcN(0, \sigma^2)}}[z \in J]} \le \eps.\]
\end{theorem}

We apply this theorem to the distributions induced by a set we aim to partition. This is our underlying
distribution $\mcD$ conditioned on a (potentially randomized) indicator, such as the indicator of a point 
being misclassified by a hypothesis; we will represent that indicator as $r(x)$. 
After projecting on the subspace orthogonal to our mean estimates, these distributions will still not have 
mean exactly zero, so we incorporate an error term into our analysis to account for that.
The resulting claim is the following:

\begin{claim}[Masses of intervals under random projections]
\label{clm:interval-mass}
Let $r(x)$ be a randomized Boolean function over $\R^d$, $P$ be a $d/2$-dimensional linear subspace in
$\R^d$, and $J \subseteq \R$ be an interval. Let $\mcP$ be a $d/2 \times d$ matrix whose rows are an orthonormal basis for $P$. Assume the following hold for some $\beta_1,
\beta_2 \in (0,1)$ and $\beta_3 \in ((d/2)^{-1/2},1)$:
\begin{itemize}
		\item \textbf{Bounded mean:} $\norm{\mcP \Ex_{x \sim \mcD}[r(x)x]}_2 \le \beta_1$.
		\item \textbf{Thin-shell support:} Every $x$ for which $r(x)$ can be nonzero satisfies
				$\frac{\norm{\mcP x}_2}{\sqrt{d/2}} = 1 \pm \beta_2$.
		\item \textbf{Large support:} $\Ex_{x \sim \mcD}[r(x)] \ge \beta_3$.
\end{itemize}
For unit vector $v$, let $\mcD_v^r$ denote the distribution of $\langle x, v\rangle, x \sim \mcD$ conditioned on $r(x) = 1$.
With probability at least 
\[1 - \frac{(1 + \beta_1/\beta_3) \ln (1/\eps)}{\eps^3 \beta_3(1 - 2\beta_2 - 2\beta_1)} \cdot \exp\Paren{-\Omega\Paren{\frac{\eps^4 \beta_3 (1 - 2\beta_2 - 2\beta_1) d}{\ln(1/\eps)(1 + \beta_1/\beta_3)}}}\]
over $v$ drawn uniformly from unit vectors in $P$, we have
\[\Abs{\Prx_{z \sim \mcD_v^r}[z \in J] - \Prx_{z \sim \mcN(0,1)}[z \in J]} \le O(\beta_1/\beta_3 + \beta_2 + \eps).\]
\end{claim}

\begin{proof}
Let $\mu \coloneqq \Ex_{x \sim \mcD, r}[x~|~r(x) = 1]$. We will define $D_{\mu,\mcP}^r$ to be the
distribution of $\mcP (x - \mu)$ and $\mcD_{\norm{\cdot}}$ to be the distribution of $\norm{\mcP(x - \mu)}_2/\sqrt{d/2}$, where $x$ is drawn from $\mcD$ conditioned on $r(x) = 1$.

We aim to apply \Cref{thm:projection-concentration} to the zero-mean distribution $\mcD_{\mu, \mcP}^r$
in the $d/2$ dimensional space $P$, 
and find the appropriate bounds for $\nu$ and $\lambda_{\max}$.
\begin{itemize}
\item \textbf{Bounding $\nu$:} By the thin-shell assumption, every $x$ for which $r(x)$ can be nonzero
satisfies $\norm{\mcP x}_2 \ge \sqrt{d/2}(1 - \beta_2)$. Thus for any such $x$, we have 
$\norm{\mcP (x-\mu)}_2 \ge \sqrt{d/2}(1 - \beta_2) - \norm{\mcP \mu}_2$ by triangle inequality.
Now we just need to bound $\norm{\mcP \mu}_2$:
\begin{align*}
\norm{\mcP \mu}_2 &= \norm{\mcP \Ex_{x, r}[x~|~r(x) = 1]}_2\\
				  &= \norm{\mcP \int_{x \in \R^n} x \cdot \frac{\mcD(x)\Prx_r[r(x)=1]}{\Prx_{x,r}[r(x)=1]} dx}_2\\
				  &= \norm{\mcP \frac{\Ex_{x, r}[r(x) x]}{\Ex_{x,r}[r(x)]}}_2\\
				  &=\frac{\norm{\mcP \Ex_{x, r}[r(x)x]}_2}{\Ex_{x,r}[r(x)]} \le \frac{\beta_1}{\beta_3}
\end{align*}
by the bounded-mean and large support assumptions.
This gives us $\nu \ge 1 - \beta_2 - \frac{\beta_1}{\beta_3 \sqrt{d/2}}$,
and thus $\nu^2 \ge 1- 2\beta_2 - 2\frac{\beta_1}{\beta_3 \sqrt{d/2}} \ge 1 - 2\beta_2 - 2\beta_1$, by the assumption that $\beta_3 \ge 1/\sqrt{d/2}$.
\item \textbf{Bounding $\lambda_{\max}$:} We want a spectral bound on $\Ex_{x \sim \mcD_{\mu,\mcP}^r}[xx^T]$. We have:
\begin{align*}
		\Ex_{x \sim \mcD_{\mu,\mcP}^r}[xx^T] &= \frac{\Ex_{x \sim \mcD, r}[r(x) (\mcP(x - \mu))(\mcP(x - \mu))^T]}{\Ex_{x,r}[r(x)]}\\
		&\le \frac{1}{\beta_3} \mcP \Ex_{x \sim \mcD}\Brac{\Paren{xx^T + \mu \mu^T - x\mu^T - \mu x^T}} \mcP^T 
\end{align*}
Since $\mcD$ has mean zero, $\mcP$ and $\E[xx^T]$ has operator norm 1, and $\mcP \mu \mu^T \mcP^T$ has operator norm $\le \norm{\mcP \mu}_2$, 
we have 
\[\lambda_{\max} = \norm{\Ex_{x \sim \mcD_{\mu, \mcP}^r}[xx^T]}_{op} \le \frac{1 + \norm{\mcP \mu}_2}{\beta_3} \le \frac{1 + \beta_1/\beta_3}{\beta_3}.\]
\end{itemize}

Substituting these parameters into \Cref{thm:projection-concentration} gives us the claimed failure probability
bound. 
The condition that holds with high probability is that for all intervals $J$, we have
\[\Abs{\Prx_{z \sim (\mcD_{\mu,\mcP}^r)_v}[z \in J] - \Prx_{\substack{\sigma \sim \mcD_{\norm{\cdot}}\\ z \sim \mcN(0, \sigma^2)}}[z \in J]} \le \eps.\]
By the thin-shell assumption, the bound on $\norm{\mcP \mu}_2$, and the triangle inequality,
$\mcD_{\norm{\cdot}}$ is supported on 
$[1 - \beta_1/\beta_3 - \beta_2, 1 + \beta_1/\beta_3 + \beta_2]$.
Then by the TV distance bound for Gaussians, we have for any $\sigma$ in the support of $\mcD_{\norm{\cdot}}$:
\[d_{TV}(\mcN(0,1), \mcN(0,\sigma^2)) \le \frac{3|1 - \sigma^2|}{2} \le O(\beta_1/\beta_3 + \beta_2),\]
so we have 
\[\Abs{\Prx_{z \sim (\mcD_{\mu,\mcP}^r)_v}[z \in J] - \Prx_{z \sim \mcN(0,1)}[z \in J]} \le O(\eps + \beta_1/\beta_3 + \beta_2).\] 
We now substitute the definition of $(\mcD_{\mu, \mcP}^r)_v$ and use the fact that $v$ is in $P$.
\begin{align*}
		\Prx_{z \sim (\mcD_{\mu,\mcP}^r)_v}[z \in J] &= \Prx_{x \sim \mcD^r}[\langle \mcP(x - \mu), v \rangle \in J]\\
		&= \Prx_{x \sim \mcD^r}[\langle x,v \rangle + \langle \mu, v \rangle \in J]\\
		&= \Prx_{y \sim \mcD^r_v}[y + \langle \mu, v \rangle \in J].
\end{align*}
So we have
\[\Abs{\Prx_{z \sim \mcD^r_v}[z + \langle \mu, v\rangle \in J] - \Prx_{z \sim \mcN(0,1)}[z \in J]} \le O(\eps + \beta_1/\beta_3 + \beta_2).\] 
By applying the uniform convergence property of \Cref{thm:projection-concentration} for all intervals to shift $J$, we have 
\[\Abs{\Prx_{z \sim \mcD^r_v}[z \in J] - \Prx_{z \sim \mcN(0,1)}[z - \langle \mu, v\rangle \in J]} \le O(\eps + \beta_1/\beta_3 + \beta_2).\] 
By anticoncentration of the Gaussian distribution and the fact that 
$\abs{\langle \mu, v \rangle} \le \beta_1/\beta_3$, we have
\[\Prx_{z \sim \mcN(0,1)}[z - \langle \mu, v \rangle \in J] = \Pr_{z \sim \mcN(0,1)}[z \in J] \pm 2\beta_1/\beta_3.\]
Combining everything, we have 
\[\Abs{\Prx_{z \sim \mcD^r_v}[z \in J] - \Prx_{z \sim \mcN(0,1)}[z \in J]} \le O(\eps + \beta_1/\beta_3 + \beta_2)\] 
with probability at least 
\[1 - \frac{(1 + \beta_1/\beta_3) \ln (1/\eps)}{\eps^3 \beta_3(1 - 2\beta_2 - 2\beta_1)} \cdot \exp\Paren{-\Omega\Paren{\frac{\eps^4 \beta_3 (1 - 2\beta_2 - 2\beta_1) d}{\ln(1/\eps)(1 + \beta_1/\beta_3)}}}\]
as desired.
\end{proof}

We now argue that for $r$ indicating a set of large Gaussian volume and 
$\mcP$ being orthogonal to a good approximation to the mean of $r$,
the assumptions of 
\Cref{clm:interval-mass} are satisfied with small values of $\beta$.

The following facts are relevant to this proof.
\begin{fact}[Thin-shell concentration of log-concave variables \cite{guédon2011}]
\label{fact:norm-concentration}
Let $X$ be an isotropic random vector with log-concave density in $\R^d$. Then there are 
universal constants $c, C > 0$ such that for all $t \ge 0$,
\[\Pr\Brac{\abs{\norm{X}_2 - \sqrt{d}} \ge t} \le C\exp(-cd^{1/2} \min((t/\sqrt{d})^3, t/\sqrt{d})).\]
\end{fact}

\begin{corollary}
\label{cor:norm-concentration}
Let $X$ be an isotropic random vector with log-concave density in $\R^d$.
Then there is a universal constant $C$ such that for any $\eps > 0$,
\[\Pr\Brac{\abs{\norm{X}_2 - \sqrt{d}} \ge C (d\log(1/\eps))^{1/3}} \le \eps.\]
\end{corollary}

\begin{fact}[Thin-shell concentration of Gaussian variables (standard fact)]
\label{fact:gaussian-concentration}
Let $X$ be a standard Gaussian in $\R^d$. Then
\[\Pr\Brac{\abs{\norm{X}_2 - \sqrt{d}} \ge t} \le 2 \exp(-t/4).\]

\end{fact}

\begin{claim}[Log-concave distributions satisfy assumptions of \Cref{clm:interval-mass}]
\label{clm:log-concave-satisfies-properties}
Let $\mcD$ be an isotropic log-concave distribution over $\R^d$ and $K$ be a sufficiently large
constant depending on those given by \Cref{fact:norm-concentration}.
Let $r$ be a randomized Boolean function such that 
$\Ex_{x \sim \mcD}[r(x)] \ge \beta$,
and let $\hat{\mu}(r)$ be such that for all $i \in [d]$,
\[\abs{\hat{\mu}(r)_i - \Paren{\Ex_{x \sim \mcD}[r(x)x]}_i} \le 1/d.\]
Let $P$ be a $d/2$-dimensional linear subspace orthogonal to $\hat{\mu}(r)$, and
$\mcP$ be a $d/2 \times d$ matrix whose rows are an orthonormal basis for $P$.
Let $r_{\trunc}$ be the truncation of $r$ defined as
\[r_{\trunc}(x) \coloneqq r(x)
\land \Ind\Brac{\frac{\norm{\mcP x}_2}{\sqrt{d/2}} = 1 \pm \frac{10K d^{1/3} \ln d}{\sqrt{d/2}}}.\] 
Then the following conditions hold:
\begin{itemize}
		\item \textbf{Bounded mean:} $\norm{\mcP \Ex_{x \sim \mcD}[r_{\trunc}(x)x]}_2 \le O(\sqrt{1/d})$.
		\item \textbf{Thin-shell support:} Every $x$ for which $r_{\trunc}(x)$ can be nonzero satisfies
		$\frac{\norm{\mcP x}_2}{\sqrt{d/2}} = 1 \pm \frac{10K d^{1/3} \ln d}{\sqrt{d/2}}$.
		\item \textbf{Large support:} $\Ex_{x \sim \mcD}[r_{\trunc}(x)] \ge \beta - O(1/d)$.
\end{itemize}
Furthermore, $\Prx_{x \sim \mcD}[r(x) \ne r_{\trunc}(x)] \le O(1/d)$.

\end{claim}

\begin{proof}
First we will prove the final inequality,
\[\Prx_{x \sim \mcD}[r(x) \ne r_{\trunc}(x)] \le O(1/d).\] 
Since the projection of $\mcD$ by $\mcP$ is a log-concave distribution over $\R^{d/2}$ (\Cref{fact: log-concave closed under projections}),
this inequality follows from
from \Cref{cor:norm-concentration} applied to $\mcP(\mcD)$ with $\eps = 1/d$.
We now prove the other conditions.
\begin{itemize}
\item 
\textbf{Bounded mean:} 
Let 
\[\mu (r_{\trunc})\coloneqq \Ex_{x \sim \mcD}[r_{\trunc}(x)x] \quad \text{and} \quad \mu(r) \coloneqq \Ex_{x \sim \mcD}[r(x)x].\]  

First we claim that for every $i$, $\abs{\mu(r_{\trunc})_i - \mu(r)_i} \le 1/d.$
We have 
\begin{align*}
\mu(r)_i &= \mu(r_{\trunc})_i +\Ex[r(x)x_i \cdot \Ind[r(x) \ne r_{\trunc}(x)]] \\
&\le \mu(r_{\trunc})_i + \Ex[\norm{x} \cdot \Ind[r(x) \ne r_{\trunc}(x)]\\
&\le \mu(r_{\trunc})_i + \int_{t = 10 K d^{1/3} \ln d}^\infty t \cdot C\exp(-ct^3/d) dt \\
&\le  \mu(r_{\trunc})_i + C\int_{t = 10 K d^{1/3} \ln d}^\infty t^3 \cdot \exp(-ct^3/d) dt \\
&\le  \mu(r_{\trunc})_i + C\int_{u = (10 K \ln d)^3 d}^\infty u \cdot \exp(-cu/d) du \\
&\le  \mu(r_{\trunc})_i + \exp(-(10K\ln d)^3 \cdot c) \cdot ((10K\ln d)^3\cdot c+1) \cdot (d/c)^2 \\
&\le \mu(r_{\trunc})_i + 1/d.\tag{by setting $K$ sufficiently large}
\end{align*}
A symmetric argument lower bounds $\mu(r)_i$ by $\mu(r_{\trunc})_i - 1/d$.
Now we bound $\norm{\mcP \mu(r_{\trunc})}_2$.
We have $\norm{\mcP \mu(r_\trunc)}_2 \le \sup_{u \in P, \norm{u}_2 = 1} \langle u, \mu(r_{\trunc})\rangle.$
Consider a unit vector $u \in P$. We have 
\begin{align*}
		\langle u, \mu(r_{\trunc})\rangle &= \langle u, \mu(r_{\trunc}) - \mu(r) \rangle+ \langle u, \mu(r) - \hat{\mu}(r)\rangle + \langle u, \hat{\mu}(r) \rangle \\
								 &\le \sum_{i \in [d]} \abs{\frac{u_i}{d}} +  \sum_{i \in [d]} \abs{\frac{u_i}{d}}  + 0   \\
								 &\le 2/\sqrt{d}.
\end{align*}
\item 
\textbf{Thin-shell support:} This follows immediately from the definition of $r_{\trunc}$.
\item 
\textbf{Large support:} This follows immediately from the assumption that
$\Ex_{x \sim \mcN(0,I_d)}[r(x)] \ge \beta$ and the fact that
$\Prx_{x \sim \mcN(0,I_d)}[r(x) \ne r_{\trunc}(x)] \le O(1/d)$.
\end{itemize}
\end{proof}

\begin{claim}[Accuracy of the mean estimates]
\label{clm:mean-estimates}
Let $S$ be a set of $10d^{3}$ points drawn from a distribution $\mcD$ with covariance
$I_d$.
Let $\mcF$ be a collection of eight (possibly randomized) Boolean functions over $\R^d$.
Then with probability at least $1-O(d^{-3})$, the following holds for all $f \in \mcF$ and $i \in [d]$:
\[\abs{\frac{1}{T} \sum_{x \in T} x_i f(x) - \Ex_{x \sim \mcD} x_i f(x)} \le 1/d.\]
\end{claim}

\begin{proof}
By the fact that $\mcD$ has covariance $I_d$, we have $\Var(x_i) = 1$ and $\Var(\frac{1}{|T|}\sum x_i)= 1/|T|$,
so by Chebyshev's inequality
\[\Prx_{T}\Brac{\abs{\Ex_{x \sim \mcD}[x_i f(x)] - \frac{1}{|T|} \sum_{x \in T} x_i f(x)}> \frac{10 d^2}{|T|}} \le 0.01d^{-4}.\]
Set $|T| \coloneqq 10 d^3$; then the deviation becomes $\frac{10d^2}{|T|} = \frac{1}{d}$.
Union bounding over $\mcF$ and $i \in [d]$, we have with probability at least $1 - O(d^{-3})$, all $f$ and all $i$
satisfy
\[\abs{\frac{1}{T} \sum_{x \in T} x_i f(x) - \Ex_{x \sim \mcD} x_i f(x)} \le 1/d.\]
\end{proof}

\begin{claim}[Accurate partitioning of (randomized) sets]
\label{clm:partition-accuracy}
Let $\mcD$ be a distribution over $\R^d \times \bits$ such that the $\R^d$-marginal is isotropic log-concave.
Let $\eps > d^{-1/7}$.
Let $\mcF$ be the set of eight (possibly randomized) Boolean functions over $\R^d$ given by $\RobInd_i$ and $\ErrInd_i(x) \coloneqq \Ind[y \ne h_i(x)]$ 
for $i \in [4]$.
Let $w_1,\ldots, w_4$ be the weights returned by \textsc{ComputeRoundingThresholds} and
$J_1, \ldots, J_4$ be the intervals generated by \textsc{ComputeClassifier}.
Then with probability at least $1-O(\delta)$, the following holds for all $f \in \mcF$ and all $i \in [4]$:
\[\Prx_{x \sim \mcD} [\langle x, u\rangle \in J_i \land f(x) = 1] = \Prx_{x \sim \mcD}[f(x)=1] \cdot w_i \pm O(\eps).\]
\end{claim}

\begin{proof}
First we will claim that each iteration succeeds with probability $1 - 1/\poly(n)$.
Consider only the functions $f \in \mcF$ such that $\Pr[f(x)=1] \ge \eps$; for the others, the statement 
holds trivially.
\textsc{ComputeClassifier} estimates the means $\mu_i, \mu_i'; i \in [4]$ of the functions from its
sample of size $10d^3$, then sets $P$ to be a linear subspace orthogonal to all of these.
By \Cref{clm:mean-estimates} we have that with probability $1 - O(d^{-3})$ the mean-accuracy assumption of 
\Cref{clm:log-concave-satisfies-properties} is satisfied.
Thus, by \Cref{clm:log-concave-satisfies-properties}, the assumptions of \Cref{clm:interval-mass} are satisfied
with the following parameters:
\begin{itemize}
\item
$\beta_1 = O(\sqrt{1/d})$
\item 
$\beta_2 = O(d^{-1/6}\ln d)$
\item 
$\beta_3 = \eps - O(1/d)$.
\end{itemize}
\textsc{ComputeClassifier} then generates a uniform random unit vector $u \in P$ and four intervals 
$J_1\ldots J_4$ of 
Gaussian volume $w_1,\ldots, w_4$. 
By \Cref{clm:interval-mass}, with high probability over $u$, the following holds for each $J_i$:
\[\Prx_{x \sim \mcD | f(x)=1} [\langle x, u\rangle \in J_i] = w_i \pm O(\beta_1/\beta_3 + \beta_2 + \eps) = w_i \pm O(\eps);\]
taking the conjunction with the event that $f(x) = 1$ gives us the desired
\[\Prx_{x \sim \mcD} [\langle x, u\rangle \in J_i \land f(x)=1] = \Prx_{x \sim \mcD}[f(x)=1]\cdot (w_i \pm O(\eps)) \le \Prx_{x \sim \mcD}[f(x)=1]\cdot w_i \pm O(\eps).\]
The failure probability is
\begin{align*}
\frac{(1 + \beta_1/\beta_3) \ln (1/\eps)}{\eps^3 \beta_3(1 - 2\beta_2 - 2\beta_1)} &\cdot \exp\Paren{-\Omega\Paren{\frac{\eps^4 \beta_3 (1 - 2\beta_2 - 2\beta_1) d}{\ln(1/\eps)(1 + \beta_1/\beta_3)}}} \\
																				   &\frac{(1 + o(1)) \ln(1/\eps)}{\eps^4 (1 - o(1))} \cdot \exp\Paren{-\Omega\Paren{\frac{\eps^5 (1 - o(1)) d}{\ln(1/\eps)(1 + o(1))}}} \\
																				   &\frac{\ln(1/\eps)}{\eps^4} \cdot \exp\Paren{-\Omega\Paren{\frac{\eps^5 d}{\ln(1/\eps)}}} \\
																				   &1/\eps^5 \cdot \exp\Paren{-\Omega(\eps^6 d)}
\end{align*}
By our assumption that $\eps > \Omega(d^{-1/7})$, this term is dominated by $\exp(-\Omega(d^{1/7}))$.
The total failure probability is dominated by the $O(d^{-3})$ probability of failure for the mean estimation.
Thus, we have that with probability $\ge 1 - 1/\poly(d)$, no bad tail events occur, and the guarantee holds
for all $f \in \mcF$ and $i \in [4]$.

The following boosting argument is relevant only when $\delta$ is smaller than this $\approx d^{-3}$ failure
probability.
In each iteration we test the accuracy of the partition by comparing
\[\Pr[\langle x, u\rangle \in J_i \land f(x) = 1]\quad \text{to} \quad \Pr[f(x)=1] \cdot w_i\]
for each function $f$.
It suffices for these estimates to be accurate up to $\eps$ additive error.
By a Chernoff bound, for a test set of size $\Omega(\log(1/\delta)/\eps^2)$, all estimates are $\eps$-accurate
with probability at least $1-\delta^2$, so that after union bounding over all iterations the
estimates are still accurate with probability at least $1 - \delta$.
Since $d\ge 2$, the probability that a good partition is not found in any of the $\log(1/\delta)$ independent 
attempts is at most $2^{-\log(1/\delta)} = \delta$.
Thus, the total failure probability is at most $2\delta$.
\end{proof}
\subsection{Correctness and complexity of \Cref{alg:ComputeClassifier}.}
\label{sec: correctness of ComputeClassifier} 
Now we will finish the analysis of \textsc{ComputeClassifier}, restated here for convenience.
\ComputeClassifier*
\begin{proof}[Proof of \Cref{thm:compute-partition}]

By \Cref{lem:robustness-lca} and the triangle inequality,
we have that all $i \in [4]$ satisfy
\begin{itemize}
\item[(i)] 
		$\Prx_{(x,y) \sim \mcD}[h_i(x) \ne y] \le \err_{\mcD}(\sign(p-t_i)) + \wh{\iso}_{\mcD, 10r}(\sign(p-t_i), 0.8)$
\item[(ii)]
$\Ex_{x \sim \mcD}[\RobInd_i(x)] \ge 1 - O(\wh{\NS}_{\mcD,10r}(\sign(p - t_i))$
\item[(iii)] 
For every $x$ such that $\RobInd_i(x)=1$ and $x':\norm{x - x'} \le r$, we have $h_i(x') = h_i(x)$.
\end{itemize}

First we will handle the error condition. 
By \Cref{thm:ComputeRoundingThresholds}, we have that with probability $1-O(\delta)$,
\begin{align*}
		\sum_{i \in [4]} w_i \cdot ( \err_{\mcD}(\sign(p-t_i) &+ \wh{\iso}_{\mcD, 10r}(\sign(p-t_i), 0.8) ) \\
&\le \Ex_{t \sim [-1,1]}[\err_{\mcD}(\sign(p-t))] + O(\eps).
\end{align*}
Thus, we have
\[\sum_{i \in [4]} w_i \cdot \Prx_{(x,y) \sim \mcD}[h_i(x) \ne y] \le \Ex_{t \sim [-1,1]}[\err_{\mcD}(\sign(p-t))] + O(\eps).\]
By \Cref{clm:partition-accuracy}, we then have with probability $\ge 1-O(\delta)$, 
\[\sum_{i \in [4]} \Prx_{(x,y) \sim \mcD}\Brac{h_i(x) \ne y \land \langle x, u^\star \rangle \in J_i} \le \Ex_{t \sim [-1,1]}[\err_{\mcD}(\sign(p-t))] + O(\eps). \]
Therefore, applying the definition of $h$, we have
\begin{align*}
\err_{\mcD}(h) \le \E_{t \sim [-1,1]}[\err_{\mcD}(\sign(p-t))] + O(\eps).
\end{align*}

Now we analyze the robustness condition. By \Cref{thm:ComputeRoundingThresholds}, we have that with probability $\ge 1-O(\delta)$,
\[\sum_{i \in [4]} w_i \cdot \wh{\NS}_{\mcD,10r}(\sign(p-t_i)) \le O(r + \eps).\]
By a Markov bound and the definitions of $\RobInd$ and $\wh{\NS}$, we have the following for each $i$:
\begin{align*}
		\Ex[\RobInd_i(x)] &= \Prx[\hat{\phi}_{\sign(p-t_i)}(x) \le 0.1] = 1 - \Pr[\hat{\phi}_{\sign(p-t_i)}(x) > 0.9] \\
				  &= 1 - \Pr\Brac{\hat{\phi}(x) > 0.9/\wh{\NS}(\sign(p-t_i)) \cdot \E[\hat{\phi}(x)]} \\
				  &\ge 1 - \lfrac{10}{9}\wh{\NS}(\sign(p-t_i).
\end{align*}
Thus we have 
\begin{align*}
		\sum_{i \in [4]} w_i \cdot \Ex_{\mcD}[\RobInd_i(x)] &\ge \sum_{i \in [4]} w_i \cdot (1 - \lfrac{10}{9} \Ex_{\mcD}[\hat{\phi}_{\sign(p-t_i),10r}(x)]) \\
															&\ge 1 - \sum_{i \in [4]} w_i \cdot \lfrac{10}{9} \Ex_{\mcD}[\hat{\phi}_{\sign(p-t_i),10r}(x)]\\
													&\ge 1 - O(r + \eps).
\end{align*}
By \Cref{clm:partition-accuracy} we then have with probability $\ge 1- \delta$,
\[\sum_{i \in [4]} \Ex_{x \sim \mcD}\Brac{\RobInd_i(x) \cdot \Ind[\langle x, u^\star \rangle \in J_i]} \ge 1 - O(r+\eps). \]
We now claim that for each $J_i$, all $x$ with $\RobInd_i(x)=1$ such that $\langle x, u^\star \rangle$ 
is at least $r$ away from the interval boundary satisfy the adversarial robustness condition
\[\forall x':\norm{x - x'}_2 \le r \text{ and } h(x) = h(x').\]
Since $x$ is $r$ away from the interval boundary, all $x'$ such that $\norm{x - x'}_2 \le r$ satisfy 
$\langle x', u^\star \rangle \in J_i$, so they are also labeled by $h_i$. By \Cref{lem:robustness-lca}, all $x$
with $\RobInd_i(x) = 1$ satisfy the robustness condition in $h_i$, so, since they are also $r$ away from the boundary, they satisfy the robustness condition in $h$.
Thus we have 
\begin{multline*}
		\AdvRob_{\mcD,r}(h) \ge \sum_{i \in [4]} \Paren{\Ex_{x \sim \mcD}\Brac{\RobInd_i(x) \cdot \Ind[\langle x, u^\star \rangle \in J_i]}} \\
		- \Prx_{x \sim \mcD}[\langle x, u^\star \rangle\text{ within $r$ of an interval boundary}]\\
\ge 1 - O(r+\eps) - O(r),
\end{multline*}
where the boundary probability is bounded by $O(r)$ due to \Cref{fact: log-concave are anticoncentrated}.
Thus when all subroutines are successful, the guarantees of the theorem hold. The total success probability is
$\ge 1 - O(\delta)$
after union bounding the failure probabilities of all subroutines.
This concludes the proof of correctness.

We now analyze the running time and query complexity to $\hat{\phi}$.
{\sc ComputeRoundingThresholds} is called once. For $\log(1/\delta)$ iterations,
it estimates $\err_t, \NS_t, \iso_t$ for each of $100/\eps^2$ rounding thresholds $t$;
this step makes $O(\log(1/\delta)/\eps^2)$ queries to $\hat{\phi}$ and takes 
$\poly(d^{k} \cdot 1/\eps \cdot \log(1/\delta))$ additional time, due to the evaluations of
the degree-$k$ polynomial $p$.
Then it solves a linear program of constant size for $(100/\eps^2)^4$ iterations. 
The total running time and query complexity of {\sc ComputeRoundingThresholds} are dominated by the first term, which is $\poly(d^{k} \cdot 1/\eps \cdot \log(1/\delta))$.

The rest of {\sc ComputeClassifier} repeats $\log(1/\delta)$ times and does the following each repetition:
\begin{itemize}
		\item[a)] evaluate $\Ind[h_i(x) \ne y]$ and $\RobInd_i(x)$ for each $i \in [4]$ and $(x,y)$ in $M$,
    \item[b)] obtain a basis orthogonal to the estimated mean vectors, a random unit vector in this space, and a partition of the real line into intervals
    \item[c)] check accuracy of the partition with respect to $S_{\test}$. 
\end{itemize}

Each evaluation of $\Ind[h_i(x) \ne y]$ makes one query to $\hat{\phi}$ and takes $d^{O(k)}$ time, as it simply calls {\sc RobustnessLCA}, which makes one evaluation of a degree-$k$ PTF and one call to $\hat{\phi}$.
Each evaluation of $\RobInd_i(x)$ makes one query to $\hat{\phi}$ and takes $O(1)$ additional time.
Thus items (a) and (c) take $d^{O(k)}$ time and queries.
For item (b), obtaining the basis takes $\poly(d)$ time by Gaussian elimination, and approximating the interval boundaries can be done in $\poly(1/\epsilon)$ time by using a numerical $\epsilon$-approximation to the error function\footnote{Such numerical approximations are standard and take  $\poly(1/\epsilon)$ time. Alternatively, one could use a randomized algorithm that takes  $\poly(1/\epsilon)$ samples from the gaussian and uses them to approximate the relevant values $w_i$. Note that we presented the $w_i$'s as exact for conciseness, but since \Cref{clm:partition-accuracy} already guarantees only an $O(\eps)$-accurate partition, we see that an additional $\eps$ error in the estimation of the $w_i$'s is asymptotically irrelevant.} 
\[
\text{erf}(t)=\frac{1}{\sqrt{\pi}}\int_{0}^{t}e^{-t^2/2}.
\]
Overall, the total time and query complexity of {\sc ComputeClassifier} is $\poly(d^{k} \cdot 1/\eps \cdot \log(1/\delta))$, as desired.
\end{proof}

\section{Verifiable robustness}
\label{sec:verifiable}
In this section we prove that under complexity assumptions, 
the robustness guarantee of our learning algorithm can be made efficiently verifiable
as discussed in \Cref{sec:intro-verifiable}.
The verifier certifies that a point satisfies the robustness condition.
We formally state this result:

\begin{corollary}[Deterministic robustness]
\label{cor:deterministic}
If $\textsf{P} = \textsf{BPP}$, then there is a learning algorithm $\mathcal{B}$ that, given access to labeled samples from a subgaussian isotropic log-concave distribution, runs in time $d^{\tilde{O}(1/\eps^2)} \cdot \log(1/\delta)$ and produces a hypothesis $h$ with the following guarantees:
\begin{itemize}
\item \textbf{Agnostic approximation:} With probability at least $1- O(\delta)$, $\Prx_{(x,y) \sim \mcD}[h(x) \ne y] \le \opt + O(\eps)$, 
				where $\opt$ is the misclassification error of the best halfspace.

\item \textbf{Verifiable robustness:} There is a verifier that runs in time $d^{\tilde{O}(1/\eps^2)}$ that takes as input a circuit $g$ and a point $x \in \R^d$ that \emph{always} rejects if 
\[\exists z:\norm{z}_2 \le r \text{ and } g(x) \ne g(x+z).\]
If $g = h$, then with probability at least $1 - O(\delta)$ over the randomness of $\mathcal{B}$, the verifier accepts with probability at least $1 - O(r + \eps)$ over $x \sim D$.
\end{itemize}
\end{corollary}
We use the following fact:
\begin{fact}[Derandomized estimation of $\hat{\phi}$]
\label{fact:derandomization}
If $\textsf{P} = \textsf{BPP}$, then there exists a deterministic algorithm running in time $d^{O(k)} \cdot \poly(1/\eps)$ that takes as input a degree-$k$ PTF $f$, a radius $r$, and an input $x \in \R^d$ and outputs an estimate $\hat{\phi}_{f,10r}(x)$ such that $\hat{\phi}_{f,10r}(x) = \phi_{f,10r}(x) \pm \eps$.
\end{fact}
\begin{proof}
Observe that there is a randomized algorithm running in $d^{O(k)}/\eps^2$ time that takes as input a threshold $t$ and decides if $\phi_f(x) \ge t$,
succeeding with probability $\ge 2/3$ whenever $\phi_f(x) > t+\eps$ or $\phi_f(x) < t - \eps$. 
This is the algorithm that samples $O(1/\eps^2)$ points $z$ from $\mcN(0,I_d)$ and 
evaluates $f(x+10r z)$ on each of them to estimate $\phi$; its analysis is a Chernoff bound.
If $\textsf{P} = \textsf{BPP}$, then there exists a deterministic algorithm for this decision problem also running in time $d^{O(k)} \cdot \poly(1/\eps)$.
Since $\phi_f(x) \in [0,1]$, we can binary search with $\log(1/\eps)$ iterations to find a $t$ such that $\phi_f(x) \in t \pm \eps$.
\end{proof}
The learning algorithm is simply $\textsc{RobustLearn}$, augmented to provide some extra information,
with the estimates of $\hat{\phi}$ provided by the deterministic estimator. 
The verifier checks that the hypothesis matches a ``template;'' 
this will prove that the unknown circuit is in fact of the form that \textsc{RobustLearn} is supposed to return, 
and for which our correctness analysis holds.
Since for any hypothesis matching the template, 
any point $x$ for which $\phi(x) \le 0.1$ satisfies the robustness condition,
the verifier will then deterministically estimate $\phi(x)$.

We define the template below:
\begin{lemma}[Hypothesis template]
\label{lem:template}
Fix a set $B$ of basis functions for the set of degree-$k$ polymomials over $\R^d$.
There is an algorithm $\textsc{Compile}$ that takes as input a $|B|$-length vector $\vec{v}$ of real-valued coefficients, real numbers $t_1,\ldots, t_4$, a unit vector $u \in \R^d$, and real numbers $c_1 < c_2 < c_3$.
For each $i \in [4]$, let the PTF $h_i$ be defined 
\[h_i(x) = \sign\Paren{\Paren{\sum_{b \in B} v_b \cdot  b(x)} - t_i}.\]
Let the intervals $J_1,\ldots,J_4$ be the partition of the real line induced by $c_1, c_2, c_3$.
The algorithm outputs a circuit computing the hypothesis
\[h(x) = \sum_{i \in 4} \textsc{RobustnessLCA}(x, h_i, r) \cdot \Ind[\langle x, u \rangle \in J_i].\]
It also outputs each of the PTFs $h_1,\ldots, h_4$.
The running time is $d^{O(k)}$.
\end{lemma}
\begin{proof}
Observe that there is a deterministic algorithm running in $d^{O(k)}$ time that takes as input all of the given parameters and the input $x$, and
evaluates $h(x)$,
using the deterministic implementation of $\hat{\phi}$ in $\textsc{RobustnessLCA}$.
Therefore there is a circuit of size $d^{O(k)}$ with the same behavior.
The algorithm $\textsc{Compile}$ takes this circuit and hardcodes all the input parameters except for $x$,
then outputs the resulting circuit.
This takes time $d^{O(k)}$ and the output is a circuit that takes $x$ as input and evaluates $h(x)$.
By the same argument, in $d^{O(k)}$ time $\textsc{Compile}$ can also output the PTFs $h_1,\ldots,h_4$.
\end{proof}

We include pseudocode of the verifier below.
For brevity we will refer to the package of data taken as input by $\textsc{Compile}$ as $\mathsf{data}$.
\begin{algorithm}
\begin{algorithmic}[1] 
\State \textbf{Input:} circuit $g$, point $x \in \R^n$, error tolerance $\eps$, hypothesis parameters $\mathsf{data}$
\State $h, h_1,h_2,h_3,h_4 \gets \textsc{Compile}(\mathsf{data})$
\If{$g \ne h$} \Return reject.
\EndIf
\State Let $u$ be the unit vector in $\mathsf{data}$ and $J_1,\ldots,J_4$ be the partition induced by the interval boundaries $c_1,c_2,c_3$ in $\mathsf{data}$.
\State Let $i$ be the interval such that $\langle x, u\rangle \in J_i$.
\If{$\hat{\phi}_{h_i, 10r}(x) > 0.1 - \eps$} \Return reject.
\EndIf
\If{$|\langle x, u \rangle - c_j| \le r$ for any of the interval boundaries $c_1,c_2, c_3$ in $\mathsf{data}$} \Return reject.
\EndIf
\State \Return accept.
\end{algorithmic}
\caption{$\textsc{Verify}(g, r, \eps,\mathsf{data})$:}
\label{alg:verify}
\end{algorithm}

\begin{proof}[Proof of \Cref{cor:deterministic}]
The learning algorithm $\mathcal{B}$ is \textsc{RobustLearn}, but with the following modification: it records its parameters in a $\mathsf{data}$ package, outputs $\textsc{Compile}(\mathsf{data})$ as its final hypothesis, and
outputs $\mathsf{data}$ as well.
By the guarantees of $\textsc{RobustLearn}$, with probability at least $1 - O(\delta)$,
$h$ satisfies the agnostic approximation guarantee.
When \textsc{Verify} is called on the output $h$ of $\mathcal{B}$, since it is the output of $\textsc{Compile}$, it always passes the first check.
By inspecting the proof of \Cref{lem:robustness-lca},
we see that in fact robustness holds for all $x$ such that $\hat{\phi}_{h_i}(x) \le 0.1 - \eps$ and $x$ is at least $r$ distance from each of the interval boundaries.
By \Cref{fact:derandomization}, we have $|\hat{\phi}_{h_i}(x) - \phi_{h_i}(x)| \le \eps$, thus the verifier rejects all points such that $\phi_{h_i}(x) > 0.1$.
Since it also rejects if $x$ is within $r$ of a boundary, all accepted points satisfy the robustness condition.
Furthermore, by inspecting the proofs of \Cref{thm:compute-partition} and \Cref{lem:robustness-lca},
we see that whenever $\textsc{RobustLearn}$ succeeds (probability $1 - O(\delta)$),
we have that at least $1 - O(r + \eps)$ fraction of points $x \sim \mcD$
satisfy $\hat{\phi}_{h_i}(x) \le 0.1 - \eps$ and are at least $r$ away from any boundary, and are thus verifiably robust.
Thus, the verifiable robustness guarantee holds.

The running time of $\textsc{Compile}$ for PTFs of degree $\tilde{O}(1/\eps^2)$ is $d^{\tilde{O}(1/\eps^2)}$ (\Cref{lem:template}).
The running time of $\textsc{Verify}$ is this plus the running time of the deterministic estimator for $\hat{\phi}$,
which is also $d^{\tilde{O}(1/\eps^2)}$ (\Cref{fact:derandomization}).
Thus the total running time is $d^{\tilde{O}(1/\eps^2)}$.
\end{proof}

\begin{remark}
We note that without the assumption that $\mathsf{P} = \mathsf{BPP}$, there is a randomized analogue of \Cref{cor:deterministic} where the user compiles the final hypothesis,
rather than verifying the one provided by \textsc{RobustLearn}.
Rather than outputting a circuit $h$, \textsc{RobustLearn} can just output $\mathsf{data}$, and the user can \textsc{Compile} it with a randomized implementation of \textsc{RobustnessLCA} 
in time $d^{\tilde{O}(1/\eps^2)} \cdot \log(1/\delta)$, in which case the soundness of the verifier holds with probability $1 - \delta$ for the compiled hypothesis.
\end{remark}

\section{Uniform convergence claims}
\label{sec: misc uniform convergence}
Finally, we will need to following observation about the uniform convergence fo empirical approximations of local noise sensitivity $\phi$. First, we need the following fact:
\begin{fact}
[\cite{vaart1997weak}, also see lecture notes \cite{lecture_notes_about_covering}]
\label{fact: epsilon nets from VC}
      A function class $\mcC$ of VC dimension $\Delta$ and every distribution $D_0$, there is an $\epsilon$-cover $H$ of $\mcC$ of size at most $\beta:=(O(1)/\epsilon)^{O(\Delta)}$. I.e. $H$ is a discrete subset of $\mcC$ of size $\beta$ and for every $f$ in $\mcC$ we have $h$ in $H$ for which
    \begin{equation*}
    \Pr_{x \sim D_0}[f(x)\neq h(x)]\leq \epsilon,\end{equation*}
\end{fact}
\begin{claim}
\label{claim: uniform convergence of phi}
    Let $\mcC$ be the class of degree-$k$ PTFs over $\R^d$, let $D$ be a probability distribution over $\R^d$ and let $\eta \in [0,1]$ be fixed. Then, for some constant $C$, if $S$ is a collection of $(d^{k}/\epsilon)^C \log 1/\delta$ i.i.d. examples from $D$, then with probability at least $1-O(\delta)$, 
    \begin{equation}
    \label{eq: goal uniform convergence for phi}
    \max_{f \in \mcC}
    \bigg \lvert
    \E_{x \sim S}[\phi_{f,\eta}(x)]
    -
    \E_{x \sim D}[\phi_{f,\eta}(x)]
    \bigg \rvert
    \leq O(\epsilon)
      \end{equation}
\end{claim}
\begin{proof}

We now use \Cref{fact: epsilon nets from VC}.  For us, $\mcC$ is the class of degree-$k$ PTFs, and we have $\Delta=d^{O(k)}$.     Taking the distribution $D_0$ to be an equal-weight mixture of (i) the distribution $D$ in the premise of this claim and (ii) the convolution of $D$ with the normal distribution $\mcN(0, \eta I_d)$ we see that for every $f$ in $\mcC$ there is $h$ in $H$ for which\footnote{Note that the error over $D_0$ is the average of the two errors in the inequalities below.}
    \begin{equation}
    \label{eq: epsilon-cover 1}
    \Pr_{x \sim D}[f(x)\neq h(x)]\leq 2\epsilon,\end{equation}
    \begin{equation}
    \label{eq: epsilon-cover 2}
    \Prx_{\substack{x \sim D\\ z \sim \mcN(0,I_d)}}[f(x+\eta z)\neq h(x+\eta z)]\leq 2\epsilon,\end{equation}
which via the definition of $\phi$ and the triangle inequality implies that
\begin{multline}
\label{eq: step 1 for phi}
\abs{\E_{x \sim D}[\phi_{f,\eta}(x)] - \E_{x \sim D}[\phi_{h,\eta}(x)]}
\leq 
\E_{x \sim D}\left[\bigg\lvert\phi_{f,\eta}(x) - \phi_{h,\eta}(x)\bigg\rvert\right]
\\
\leq
2\E_{x \sim D}[|f(x)- h(x)|]
+
\E_{\substack{x \sim D\\ z \sim \mcN(0,I_d)}}[|f(x+\eta z)- h(x+\eta z)|]
\leq
O(\epsilon).
\end{multline}
    
  By the standard Hoeffding bound, since $\phi$ and $\xi$ are always in $[0,1]$, with probability $1-\delta/2$ for every $h$ in $H$ we have
    \begin{equation}
        \label{eq: step 2 for phi}
    \bigg\lvert
     \E_{x \sim S}[\phi_{h,\eta}(x)]
    -
    \E_{x \sim D}[\phi_{h,\eta}(x)]
    \bigg\rvert
    \leq
    \frac{1}{\sqrt{|S|}} O(\log(|H|/\delta)) \leq \eps,
        \end{equation}
    where the last step for both inequalities follows by substituting $|H|$, $|S|$, $\Delta$ and taking $C$ to be a sufficiently large absolute constant. We also know that with for a sufficiently large absolute constant $C$, with probability $1-\delta$ the set $S$ satisfies \Cref{fact:VC-convergence for distances}, which means that:
    \begin{equation}
    \label{eq: uniform disagreement}
\sup_{f_1, f_2 \in \mcC}\left | \Pr_{x\sim S}[f_1(x)\neq f_2(x)] - \Pr_{x\sim D}[f_1(x)\neq f_2(x)] \right | \le O(\eps).
\end{equation}
From the above, we can also conclude that with probability $1-\delta$ over $S$ for every pair $f_1, f_2$ in $\mcC$ we have the following:
\begin{multline}
\label{eq: uniform disagreement smoothed}
\bigg \lvert
    \Pr_{\substack{x \sim S\\ z \sim \mcN(0,I_d)}}[f_1(x+\eta z)\neq f_2(x+\eta z)]
    -\Pr_{\substack{x \sim D\\ z \sim \mcN(0,I_d)}}[f_1(x+\eta z)\neq f_2(x+\eta z)] \bigg \rvert
    \leq \\
    \E_{z \sim \mcN(0,I_d)}
    \underbrace{\bigg \lvert
    \Pr_{x \sim S}[f_1(x+\eta z)\neq f_2(x+\eta z)]
    -\Pr_{x \sim D}[f_1(x+\eta z)\neq f_2(x+\eta z)]
    \bigg \rvert}_{  \substack{\leq O(\epsilon) \text{ by defining $f_3(x):=f_1(x+\eta z)$, $f_4(x):=f_2(x+\eta z)$},\\\text{observing $f_3$ and $f_4$ are also degree-$k$ PTFs and using \Cref{eq: uniform disagreement}}}
    }
    \leq O(\epsilon)
\end{multline}

Using the definition of $\phi$, the triangle inequality, and Equations \ref{eq: uniform disagreement} and \ref{eq: uniform disagreement smoothed} we also see that 
\begin{multline}
\label{eq: step 3 for phi}
\abs{\E_{x \sim S}[\phi_{h,\eta}(x)] - \E_{x \sim S}[\phi_{f,\eta}(x)]}
\leq 
\\
2\Pr_{x \sim S}[h(x) \neq f(x)]
+
\Pr_{\substack{x \sim S\\ z \sim \mcN(0,I_d)}}[h(x+\eta z)\neq f(x+\eta z)]
\leq\\
2\Pr_{x \sim D}[h(x) \neq f(x)]
+
\Pr_{\substack{x \sim D\\ z \sim \mcN(0,I_d)}}[h(x+\eta z)\neq f(x+\eta z)] \leq 
O(\epsilon),
\end{multline}
where in the last step we substituted \Cref{eq: epsilon-cover 1} and \Cref{eq: epsilon-cover 2}.
Finally, we put all the inequalities together. By combining the triangle inequality with Equations \ref{eq: step 1 for phi},  \ref{eq: step 2 for phi} and \ref{eq: step 3 for phi} we derive \Cref{eq: goal uniform convergence for phi}.
\end{proof}

\begin{claim}
\label{claim: uniform convergence of isolation probability}
   Let $\mcC$ be the class of degree-$k$ PTFs over $\R^d$, let $D$ be a probability distribution over $\R^d$ and let $\eta \in [0,1]$ be fixed. Then, for some constant $C$, if $S$ is a collection of $(d^{k}/\epsilon)^C \log 1/\delta$ i.i.d. examples from $D$, then with probability at least $1-O(\delta)$, 
 for every polynomial $p$ it holds that
\begin{equation}
    \label{eq: isolation probability uniform concentration new}
 \underset{x \sim D }{\Pr}\left[
\E_{t \in [-1,1]}[\phi_{\sign(p(x)-t),\eta}(x)] \ge 0.7\right]
\leq \underset{x \sim D}{\Pr}[\E_{t \in [-1,1]}[\phi_{\sign(p-t),\eta}(x)] \ge 0.67] +O(\epsilon)
\end{equation}
\end{claim}
\begin{proof}
    Let $\mcC^{\text{clipped}}$ be the class of degree-$k$ clipped polynomials, i.e. functions $p^{\text{clipped}}:\R^d \rightarrow [-1,1]$ of the form
    \[
    p^{\text{clipped}}(x)=
    \begin{cases}
    1 &\text{if }p(x)\geq 1\\
    -1 &\text{if }p(x)\leq -1\\
    p(x) &\text{otherwise}
    \end{cases}
    \]
    where $p$ is a degree-$k$ polynomial. We first argue the following
    \begin{observation}
        For any distribution $D_0$ over $\R^d$, there is a subset $H$ of $\mcC^{\text{clipped}}$ of size $(O(1)/\epsilon)^{O(d^k/\epsilon)}$, i.e. such that  for every $p^{\text{clipped}}$ in $\mcC^{\text{clipped}}$ we have
    \begin{equation}
    \label{eq: epsilon net for clipped polys}
    \min_{h \in H}[\E_{x \sim D_0}[\abs{p^{\text{clipped}}(x)-h(x)}]]\leq O(\epsilon).\end{equation}
    \end{observation}
\begin{proof} 
We form $H$ by considering an $\epsilon^2$-net $H_{\text{PTF}}^{\epsilon^2}$
for degree-$2k$ PTFs, and taking $H$ to consist of functions of the form
\[
h(x)
=
\sum_{\tau \in \{-1,-1+\epsilon,-1+2\epsilon,...,+1\}}
\tau\cdot
f_{\tau}(x),
\]
where each $f_{\tau}(x)$ is some function in $f_{\tau}(x)$. By \Cref{fact: epsilon nets from VC}, we see that $H_{\text{PTF}}^{\epsilon^2}$ has size $(O(1)/\epsilon)^{O(d^k)}$, and therefore (by a counting argument) the size of $H$ defined above is indeed $(O(1)/\epsilon)^{O(d^k/\epsilon)}$.

We can write for every $x$ and $p^{\text{clipped}}$ in $\mcC^{\text{clipped}}$ the following inequality:
\begin{multline}
\label{sandwiching for clipped polynomials}
\sum_{\tau \in \{-1,-1+\epsilon,-1+2\epsilon,...,+1\}}
\tau\cdot \Ind[p^{\text{clipped}}(x)\in(\tau, \tau+\epsilon]]
\leq
p^{\text{clipped}}(x)
\leq\\
\sum_{\tau \in \{-1,-1+\epsilon,-1+2\epsilon,...,+1\}}
(\tau+\epsilon)\cdot \Ind[p^{\text{clipped}}(x)\in(\tau, \tau+\epsilon]],
\end{multline}
we observe that each indicator $\Ind[p^{\text{clipped}}(x)\in(\tau, \tau+\epsilon]]$ equals to a degree-$2k$ polynomial threshold function, and therefore there is some $f_{\tau}$ in $H_{\text{PTF}}^{\epsilon^2}$ for which
\[
\E_{x \sim D_0}
\abs{
\Ind[p^{\text{clipped}}(x)\in(\tau, \tau+\epsilon]]
-
f_{\tau}
}
\leq O(\epsilon^2),
\]
which substituted into \Cref{sandwiching for clipped polynomials} allows us to conclude that 
\[
\E_{x \sim D_0}
\abs{
p^{\text{clipped}}(x)
-
\sum_{\tau \in \{-1,-1+\epsilon,-1+2\epsilon,...,+1\}}
\tau\cdot
f_{\tau}(x)
}
\leq O(\epsilon),
\]
which concludes the proof that for $p^{\text{clipped}}$ in $\mcC^{\text{clipped}}$ we have $h:\R^{d}$ in $H$ for which \Cref{eq: epsilon net for clipped polys} holds.

We finish the proof of the observation, by resolving one last issue: as defined the set $H$ is not a subset of  $\mcC^{\text{clipped}}$. However, if we consider the set $H'$ consisting of functions $f$ of the form
\[
H'=
\{
f=\argmin_{g \in \mcC^{\text{clipped}}}[\E_{x \sim D_0}[\abs{g(x)-h(x)}]]
: h \in H,\}
\]
then we see by the triangle inequality that $H'$ still satisfies \Cref{eq: epsilon net for clipped polys}, has a size of at most $|H|$ and $H'$ is a subset of $\mcC^{\text{clipped}}$
\end{proof}

We continue the proof of \Cref{claim: uniform convergence of isolation probability}
We can wite:
\begin{multline}
\label{eq: definition of phi for clipped polynomials}
    \E_{t \in [-1,1]}[\phi_{\sign(p(x)-t),\eta}(x)]
    =
    \Pr_{t \in [-1,1], z \sim \mcN(0,I_d)}\left[\abs{\sign(p(x)-t) \neq \sign(p(x+\eta z)-t)}\right] =\\
    \E_{ z \sim \mcN(0,I_d)}\left[\frac{\abs{p^{\text{clipped}}(x)- p^{\text{clipped}}(x+\eta z)}}{2}\right]:=\phi_{p^{\text{clipped}},\eta}(x)
\end{multline}
We will overload the notation for $\phi$ and define  $\phi_{p^{\text{clipped}},\eta}(x)$ to be the expression above. We further
define the following auxiliary quantity:
\[
\xi_{p^{\text{clipped}},\eta}(x)
=
\begin{cases}
    0 &\text{ if } \phi_{p^{\text{clipped}},\eta}(x)\leq 0.67\\
    1 &\text{ if } \phi_{p^{\text{clipped}},\eta}(x)\geq 0.7\\
    \frac{100}{3} (\phi_{p^{\text{clipped}},\eta}(x)-0.67) &\text{ if }\phi_{f,\eta}(x)\in(0.67,0.7)
\end{cases}
\]
By construction, the function satisfies the following properties for every $x$, 
\begin{equation}
\label{eq: xi is sandwiching new}
    \Ind[\phi_{p^{\text{clipped}},\eta}(x) \ge 0.67] \leq \xi_{p^{\text{clipped}},\eta}(x)
    \leq \Ind[\phi_{p^{\text{clipped}},\eta}(x) \ge 0.7]
\end{equation}
\begin{equation}
\label{eq: xi is continuous new}
    \abs{\xi_{p_1^{\text{clipped}},\eta}(x)-\xi_{p_2^{\text{clipped}},\eta}(x)}
    \leq
    O(1)\cdot \abs{\phi_{p_1^{\text{clipped}},\eta}(x)-\phi_{p_2^{\text{clipped}},\eta}(x)}
\end{equation}


   Taking the distribution $D_0$ to be an equal-weight mixture of (i) the distribution $D$ in the premise of this claim and (ii) the convolution of $D$ with the normal distribution $\mcN(0, \eta I_d)$ we see that for every $p^{\text{clipped}}$ in $\mcC^{\text{clipped}}$ there is $h$ in $H$ for which\footnote{Note that the error over $D_0$ is the average of the two errors in the inequalities below.}
    \begin{equation}
    \label{eq: epsilon-cover 1 new}
    \Pr_{x \sim D}[\abs{p^{\text{clipped}}(x)-h(x)}]\leq 2\epsilon,\end{equation}
    \begin{equation}
    \label{eq: epsilon-cover 2 new}
    \Pr_{\substack{x \sim D\\ z \sim \mcN(0,I_d)}}[\abs{p^{\text{clipped}}(x+\eta z) - h(x+\eta z)}]\leq 2\epsilon,\end{equation}
which via the definition of $\phi$ for clipped polynomials and the triangle inequality implies that
\begin{multline}
\label{eq: step 1 for phi new}
\abs{\E_{x \sim D}[\phi_{p^{\text{clipped}},\eta}(x)] - \E_{x \sim D}[\phi_{h,\eta}(x)]}
\leq 
\E_{x \sim D}\left[\bigg\lvert\phi_{p^{\text{clipped}},\eta}(x) - \phi_{h,\eta}(x)\bigg\rvert\right]
\\
\leq
2\E_{x \sim D}[|p^{\text{clipped}}(x)- h(x)|]
+
\E_{\substack{x \sim D\\ z \sim \mcN(0,I_d)}}[|p^{\text{clipped}}(x+\eta z)- h(x+\eta z)|]
\leq
O(\epsilon).
\end{multline}
Together with \Cref{eq: xi is continuous new}, this implies that
\begin{equation}
\label{eq: step 1 for psi new}
\abs{\E_{x \sim D}[\xi_{p^{\text{clipped}},\eta}(x)] - \E_{x \sim D}[\xi_{h,\eta}(x)]}\leq O(1)\cdot  
\E_{x \sim D}\left[\bigg\lvert\phi_{p^{\text{clipped}},\eta}(x) - \phi_{h,\eta}(x)\bigg\rvert\right]
\leq
O(\epsilon).
\end{equation}
    
  By the standard Hoeffding bound, since $\phi$ and $\xi$ are always in $[0,1]$, with probability $1-\delta/2$ for every $h$ in $H$ we have
    \begin{equation}
        \label{eq: step 2 for phi new}
    \bigg\lvert
     \E_{x \sim S}[\phi_{h,\eta}(x)]
    -
    \E_{x \sim D}[\phi_{h,\eta}(x)]
    \bigg\rvert
    \leq
    \frac{1}{\sqrt{|S|}} O(\log(|H|/\delta)) \leq \eps,
        \end{equation}
\begin{equation}
        \label{eq: step 2 for psi new}
    \bigg\lvert
     \E_{x \sim S}[\xi_{h,\eta}(x)]
    -
    \E_{x \sim D}[\xi_{h,\eta}(x)]
    \bigg\rvert
    \leq
    \frac{1}{\sqrt{|S|}} O(\log(|H|/\delta)) \leq \eps,
        \end{equation}
    where the last step for both inequalities follows by substituting $|H|$, $|S|$, $\Delta$ and taking $C$ to be a sufficiently large absolute constant. 

We observe that \Cref{fact:VC-convergence for distances} and \Cref{sandwiching for clipped polynomials} together imply that\footnote{
\Cref{sandwiching for clipped polynomials} tells us that the difference $p_1-p_2$ can be $\epsilon$-approximated in $L_{\infty}$ norm by a function of the form $\sum_{\tau \in \{-2,-2+\epsilon,-1+2\epsilon,...,+2\}}
\tau\cdot
f_{\tau}(x)$, where each $f_{\tau}$ is a degree-$2k$ PTF. Then, \Cref{fact:VC-convergence for distances} together with triangle inequality tells us that the inequality above holds when $C$ is a sufficiently large absolute constant.
}:
    \begin{equation}
    \label{eq: uniform disagreement new}
\sup_{p_1^{\text{clipped}}, p_2^{\text{clipped}} \in \mcC^{\text{clipped}}}\left | \E_{x\sim S}[\abs{p_1^{\text{clipped}}(x) - p_2^{\text{clipped}}(x)}] - \E_{x\sim D}[\abs{p_1^{\text{clipped}}(x)-p_2^{\text{clipped}}(x)}] \right | \le O(\eps).
\end{equation}
From the above, we can also conclude that with probability $1-\delta$ over $S$ for every pair $p_1^{\text{clipped}}, p_2^{\text{clipped}}$ in $\mcC^{\text{clipped}}$ we have the following:
\begin{multline}
\label{eq: uniform disagreement smoothed new}
\bigg \lvert
    \E_{\substack{x \sim S\\ z \sim \mcN(0,I_d)}}[\abs{p_1^{\text{clipped}}(x+\eta z)- p_1^{\text{clipped}}(x+\eta z)}]
    -\E_{\substack{x \sim D\\ z \sim \mcN(0,I_d)}}[\abs{p_1^{\text{clipped}}(x+\eta z)- p_2^{\text{clipped}}(x+\eta z)}] \bigg \rvert
    \leq \\
    \E_{z \sim \mcN(0,I_d)}
    \underbrace{\bigg \lvert
    \E_{x \sim S}[\abs{p_1^{\text{clipped}}(x+\eta z)-p_2^{\text{clipped}}(x+\eta z)}]
    -\E_{x \sim D}[\abs{p_1^{\text{clipped}}(x+\eta z)-  p_2^{\text{clipped}}(x+\eta z)}]
    \bigg \rvert}_{  \substack{\leq O(\epsilon) \text{ by defining $p_3^{\text{clipped}}(x):=p_1^{\text{clipped}}(x+\eta z)$, $p_4^{\text{clipped}}(x):=p_2^{\text{clipped}}(x+\eta z)$},\\\text{ and using \Cref{eq: uniform disagreement new}}}
    }
    \leq O(\epsilon)
\end{multline}

We now come back to the setting of \Cref{eq: epsilon-cover 1 new}. $p^{\text{clipped}}$ is a function in $\mcC^{\text{clipped}}$ and $h$ in $H$ satisfies \Cref{eq: epsilon-cover 1 new}.
Using the definition of $\phi$, the triangle inequality, and Equations \ref{eq: uniform disagreement new} and \ref{eq: uniform disagreement smoothed new} we also see that 
\begin{multline}
\label{eq: step 3 for phi new}
\abs{\E_{x \sim S}[\phi_{h,\eta}(x)] - \E_{x \sim S}[\phi_{p^{\text{clipped}},\eta}(x)]}
\leq 
\\
2\Pr_{x \sim S}[\abs{h(x) - p^{\text{clipped}}(x)}]
+
\Pr_{\substack{x \sim S\\ z \sim \mcN(0,I_d)}}[\abs{h(x+\eta z)-p^{\text{clipped}}(x+\eta z)}]
\leq\\
2\Pr_{x \sim D}[h(x) \neq f(x)]
+
\Pr_{\substack{x \sim D\\ z \sim \mcN(0,I_d)}}[h(x+\eta z)\neq f(x+\eta z)] \leq 
O(\epsilon),
\end{multline}
where in the last step we substituted \Cref{eq: epsilon-cover 1 new} and \Cref{eq: epsilon-cover 2 new}.
By Equation \ref{eq: xi is continuous new}, we see that \begin{equation}
\label{eq: step 3 for psi new}
\abs{\E_{x \sim S}[\xi_{h,\eta}(x)] - \E_{x \sim S}[\xi_{f,\eta}(x)]}
\leq 
O(1)\cdot
\abs{\E_{x \sim S}[\phi_{h,\eta}(x)] - \E_{x \sim S}[\phi_{f,\eta}(x)]} \leq 
O(\epsilon),
\end{equation}

Analogously, by combining the triangle inequality with Equations \ref{eq: step 1 for psi new},  \ref{eq: step 2 for psi new} and \ref{eq: step 3 for psi new} we derive the following:
   \begin{equation}
    \label{eq: goal uniform convergence for xi new}
    \max_{f \in \mcC}
    \bigg \lvert
    \E_{x \sim S}[\xi_{f,\eta}(x)]
    -
    \E_{x \sim D}[\xi_{f,\eta}(x)]
    \bigg \rvert
    \leq O(\epsilon),
      \end{equation}
      which together with \Cref{eq: xi is sandwiching new} and \Cref{eq: definition of phi for clipped polynomials} implies \Cref{eq: isolation probability uniform concentration new}.

\end{proof}

The following claims follow from VC theory (see e.g. \cite{van2009note} and the references therein):
\begin{fact}[Generalization bound from VC dimension]
\label{fact:VC-convergence for errors}
Let $\mcC$ be a concept class and $\mcD$ be a distribution over $\R^d$. For some sufficiently large absolute constant $C$, the following is true. With probability at least $1-\delta$ over a sample
$S$ of i.i.d. samples from $\mcD$, with $|S| \ge C \cdot VC(\mcC) \log(1/\delta)/\eps^2$
%
the following holds:
\[\sup_{f \in \mcC}\left | \err_{\mcD}(f) - \wh{\err}_{S}(f) \right | \le O(\eps).\]
\end{fact}

\begin{fact}
\label{fact:VC-convergence for distances}
Let $\mcC$ be a concept class with VC dimension $\Delta$ and $D$ be a distribution over $\R^d$. For a sufficiently large absolute constant $C$, let $S$
be a collection of $\frac{C\Delta}{\epsilon^2}\log(1/\delta)$ i.i.d. examples from $D$. Then, the following holds with probability at least $1-\delta$
\begin{equation}
\label{eq: disagreement from VC dimension}
\sup_{f_1, f_2 \in \mcC}\left | \Pr_{x\sim S}[f_1(x)\neq f_2(x)] - \Pr_{x\sim D}[f_1(x)\neq f_2(x)] \right | \le O(\eps).
\end{equation}
\end{fact}

\end{document}